\documentclass{article}

\PassOptionsToPackage{numbers, compress}{natbib}
\usepackage{doi}
    \usepackage[final]{neurips_data_2023}
\usepackage[utf8]{inputenc} % allow utf-8 input
\usepackage[T1]{fontenc}    % use 8-bit T1 fonts
\usepackage{hyperref}       % hyperlinks

\usepackage{url}            % simple URL typesetting
\usepackage{booktabs}       % professional-quality tables
\usepackage{amsfonts}       % blackboard math symbols
\usepackage{nicefrac}       % compact symbols for 1/2, etc.
\usepackage{microtype}      % microtypography
\usepackage[dvipsnames]{xcolor}         % colors
\hypersetup{
  colorlinks   = True, %Colours links instead of ugly boxes
  urlcolor     = NavyBlue, %Colour for external hyperlinks
  linkcolor    = NavyBlue, %Colour of internal links
  citecolor   = BrickRed %Colour of citations
}

\usepackage{comment}

\usepackage{multirow}       %Dhruv's Multirow Package (add to main tex)
\usepackage{caption}        %for PCA with multiple sub image rows reference
\usepackage{subcaption}     %for PCA with multiple sub image rows 
\usepackage{tikz}
\usepackage{graphicx}
\usepackage{xcolor}

\usepackage{amsthm}
\usepackage{amsmath}
\usepackage{amssymb}
\usepackage{algorithm} 
\usepackage{algpseudocode}
\usepackage{wrapfig}

\theoremstyle{definition}
\newtheorem{definition}{Definition}[section]
\newtheorem{theorem}{Theorem}[section]
\newtheorem{lemma}[theorem]{Lemma}

\algnewcommand\algorithmicforeach{\textbf{for each}}
\algdef{S}[FOR]{ForEach}[1]{\algorithmicforeach\ #1\ \algorithmicdo}

\title{Diverse Community Data for Benchmarking Data Privacy Algorithms}

\author{%
  Aniruddha Sen  \\
  University of Massachusetts \\
  Amherst, MA 01003 \\
  \And
  Christine Task \\
  Knexus Research\\
  Oxon Hill, MD 20745 \\
  \texttt{christine.task@knexusresearch.com} \\
  \And
  Dhruv Kapur \\
  University of Michigan \\
  Ann Arbor, MI 48109  \\
  \And
  Gary Howarth \\
  NIST \\
  Boulder, CO 80305 \\
  \texttt{gary.howarth@nist.gov} \thanks{corresponding author}
  \And
  Karan Bhagat \\
  Knexus Research\\
  Oxon Hill, MD 20745 \\  
}

\begin{document}

\maketitle

\begin{abstract}

The Collaborative Research Cycle (CRC) is a National Institute of Standards and Technology (NIST) benchmarking program intended to strengthen understanding of tabular data deidentification technologies. Deidentification algorithms are vulnerable to the same bias and privacy issues that impact other data analytics and machine learning applications, and it can even amplify those issues by contaminating downstream applications. This paper summarizes four CRC contributions: theoretical work on the relationship between diverse populations and challenges for equitable deidentification; public benchmark data focused on diverse populations and challenging features; a comprehensive open source suite of evaluation metrology for deidentified datasets; and an archive of more than 450 deidentified data samples from a broad range of techniques. The initial set of evaluation results demonstrate the value of the CRC tools for investigations in this field.

% \textcolor{red}{reduce size of abstract to reclaim page space}

\end{abstract}

% ------------------------------------------------------
%  Introduction
% ------------------------------------------------------

\section{Introduction}  %Gary and Christine 

Deidentification algorithms take records linked to individuals and attempt to produce data that is dissociated from individuals but remains useful for analysis. Effective deidentification permits organizations to safely share useful information from potentially sensitive data. 
Such data can be used to train machine learning algorithms; expose fraud, waste, and abuse; improve health outcomes; and much more. Many approaches are available to deidentify data. Some approaches, such as statistical disclosure control, redact or suppress information that is deemed particularly identifying. 
Other approaches, such as synthetic data algorithms leverage generative models to reproduce sensitive data distributions using new, synthetic records. Differential privacy strengthens deidentification with a rigorous mathematical definition of privacy, producing synthetic data with quantifiable bounds on the influence of any real individual.  While deidentification release algorithms may (but not necessarily) improve privacy, they can also distort data  by introducing artifacts and bias. Identifying and resolving these issues is important, but it is not trivial to do.

Intuitively, individuals with many unusual values in the data are easier to identify than those with relatively common demographic characteristics. 
Consider data with two subpopulations of equal size: group A, who are internally homogeneous and have highly-correlated features; and group B, who are internally heterogeneous with relatively independent features. 
We can infer that group B will have more uniquely identifiable individuals. 
Heavy handed deidentification will eliminate or perturb group B to a greater extent than A, at the likely expense of reducing or meaningfully altering the representation of group B. 
We use \emph{subpopulation dispersal} to refer to the sparsification of subpopulations, such as that caused by increased feature independence. 
We believe work on benchmarks containing subpopulation dispersal is imperative to overcoming these risks. Further, we advocate for evaluation metrics that explore the effects of subpopulation dispersal and other challenges induced by real-world data. 

In this work we introduce the \href{https://github.com/usnistgov/privacy_collaborative_research_cycle}{Collaborative Research Cycle (CRC)} a benchmarking effort to compare deidentification methods hosted by the National Institute of Standards and Technology (NIST). The CRC includes target data, evaluation metrics, and a repository containing community-created deidentified data and their evaluation results. Our benchmark dataset contains examples of challenging, real-world conditions, such as  subpopulation dispersal (Section \ref{sec:data}). Together this program supports unprecedented rigorous exploration of deidentification algorithm behavior. 

% We specifically recognize the dangers posed by bias introduced into deidentified data, so we have specifically created data that are challenging to deidentify and likely lead to biased outputs. 
% We provide a formal analysis of subpopulations with varying levels of heterogeneity to support our claim that our benchmark data are in fact challenging to deidenitfy. 

There are several existing synthetic data evaluation libraries, such as \href{https://docs.sdv.dev/sdmetrics/}{SDMetrics} and \href{https://github.com/ydataai}{YData} (see Table \ref{tab:synth_data_tools} for a larger selection). 
These tools are focused on evaluating deidentified data with arbitrary schema and providing application-specific feedback. 
Our effort distinguishes itself by providing specific benchmarking data and a venue to contribute directly comparable samples of  deidentification algorithm outputs and their evaluation results. 
We believe that shared benchmark data alongside common evaluation metrics promote understanding and exploration of a problem by providing common resources, vocabulary, and analytic framework. Thus, we have created the CRC as an arena to test, compare, and discuss varying approaches on equal footing.

We begin with a formal analysis of subpopulation dispersal. This concept underlies one source of tension between diversity, equity, and privacy that impacts all deidentification algorithms operating in real-world conditions. Specifically we show that subpopulations with greater feature independence leads to smaller cell counts in tabular data. We believe benchmark tabular data for deidentification technologies must present subpopulation dispersal to provide a realistically challenging target. 

Next, we introduce a benchmark dataset of real demographic data, the \href{https://doi.org/10.18434/mds2-2895}{NIST Diverse Communities Data Excerpts} \cite{task_nist_2022} (the Excerpts). The data are curated from the American Community Survey from the U.S. Census Bureau with 24 features and three geographic samples. We introduce a suite of benchmark software, the \href{https://doi.org/10.18434/mds2-2943}{SDNist Deidentified Data Report Tool} \cite{task_sdnist_2023}, which provides metrology and visualization functions comparing groundtruth and deidentified data. We use SDNist to demonstrate the  presence of subpopulation dispersal within the Excerpts. 

Finally, we introduce the \href{https://github.com/usnistgov/privacy_collaborative_research_cycle}{CRC Data and Metrics Archive}, a repository of crowd-sourced, deidentified instances of the Excerpts, all of which are benchmarked with SDNist. Members of industry, academia, and government have contributed more than 450 fully-evaluated entries to the repository, spanning many different deidentification approaches including differentially private techniques \cite{Dwork2006}, cell suppression based $k$-anonymity\cite{sweeney_k-anonymity_2002}, GAN-based synthetic data, and many others. We use SDNist and the archive to illustrate the impact of subpopulation dispersal on deidentification performance. 

The CRC is a comprehensive effort to equip algorithm developers, data owners, and theoreticians the tools to evaluate and compare tabular deidentification technologies. All of the data (ground truth and deidentified instances) and evaluation software are open-sourced to allow transparent analysis. Further, the archived data has extensive, standardized metadata to facilitate parsing and comparing approaches, such as by feature subset, algorithm type, etc.

\label{diversity}%Aniruddha and Christine 

\section{Distributional Diversity and Subpopulation Dispersal}
\label{sec:math}
%Aniruddha and Christine 

We use the term \emph{diverse populations} to refer to populations containing two or more subpopulations which differ from each other in terms of their feature correlations. By feature correlation we mean the predictive power of a feature value on other feature values for a given record. Intuitively, weaker feature correlations are more challenging for deidentification algorithms to model accurately. Here, we provide a robust formal analysis of these dynamics which helped inspire the Diverse Community Excerpts benchmark data.  
% The Diverse Community Excerpts focus on geographical regions with diverse populations to support a deeper understanding of algorithm behavior in these contexts. 
%We begin by introducing the concept of \emph{dispersal rates} and its relation to diversity and privacy.  We include real-world illustrations from the Diverse Community Excerpts data.  In section [?], we use the information theoretic concept of uncertainty (related to feature independence) to derive numerical bounds on the impact that differing distributions have on dispersal rates.  
% These differences pose challenges for maintaining equity in deidentified data. 
%In section [?] we use submissions from NIST CRC DeID Data and Metrics Archive to illustrate how different approaches respond to diversity in practice.  

%\subsection{Defining Dispersal Rates and Correlations}

%[Consider a population $P$, a schema $S$, and a histogram defined as follows.  For example, if there are 10 people and two features we might see something like this.  We define dispersal rate as follows. Intuitively this is the proportion of the population that ].

 Consider data represented in a histogram, with bins counting the number of occurrences of each record type (formally defined below). Individuals in bins with small counts are more uniquely identifiable--there are fewer people like them. Many privacy approaches focus on protecting these individuals.  Traditional statistical disclosure control techniques like $k$-anonymity operate by perturbing or redacting records in small bins. Randomization approaches that rely on additive noise, such as differential privacy or subsampling, have much larger relative impact on these small bins. Non-differentially private synthetic data generators may have difficulty modeling sparser areas of the distribution. Thus, deidentified data fidelity is more challenging for individuals in small bins. However these individuals are not necessarily unimportant outliers. Depending on feature independence it is possible for a large subpopulation to become dispersed into small bins.    

In diverse data the same schema may disperse one subpopulation and preserve the other, causing the first to potentially be erased by deidentification while the maintaining the second (see Figure \ref{fig:orange_blue}).  Therefore, it is vital to use diverse data to study algorithm behavior.  We first introduce two relevant information theoretic tools and our notation. We then formally define \emph{dispersal ratio} and derive bounds on the relationship between diversity and dispersal.  

\begin{figure*}[t]
\centering
    % \begin{subfigure}[t]{\textwidth}
        \centering
  % include first image
        \includegraphics[width=\linewidth]{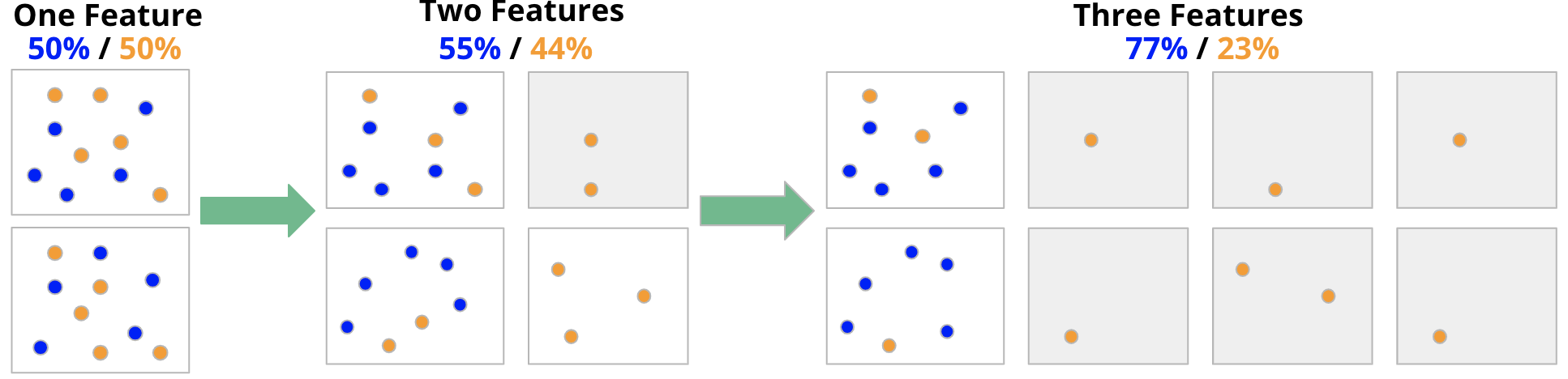}  
        \caption{Boxes represent partition of observations by an increasing number of categorical features. Blue group features are highly correlated and orange group features are more independent. The orange group is, therefore, dispersed into small-count bins (in grey). The percentage label shows the relative representation of each group if small count bins were suppressed to protect privacy.}
        \label{fig:orange_blue}
    % \end{subfigure}
    \end{figure*}

\begin{definition}
[Entropy] The entropy of a discrete random variable $X$ is defined as:
$${\displaystyle\mathrm {H}(X)=-\sum _{x\in {\mathcal {X}}}p(x)\log p(x)}$$ 
where $\mathcal {X}$ denotes the range of $X$ \cite{cover1999elements}. For this paper, we will be considering the empirical definition of entropy or observed entropy from the probability distribution corresponding to a histogram.
\end{definition}

%--Definition: uncertainty coefficient  
\begin{definition}
[Uncertainty Coefficient] The uncertainty coefficient $U(X|F)$ is an information theoretic metric for quantifying the entropy ($H$) correlation between two random variables $X$ and $F$ (in our case an existing feature set $F$ and a new added feature $X$) \cite{uncertaintycoeff}. It is a normalized form of mutual information. Note that $0 \le U(X|F) \le 1$. It is defined as:
\[U(X|F) = {\frac{H(X)-H(X|F)}{H(X)}}\]
An uncertainty coefficient of 1 implies that the random variable $X$ can be completely predicted by knowing the value of $F$, and an uncertainty coefficient of 0 implies that the random variable $X$ is independent of $F$. Thus, $U(X|F)$ inversely correlates with the independence of $X$ for some $F$. Note that this metric is asymmetric and may not be be equal if the variables are exchanged, i.e. $U(X|F) \neq U(F|X)$.
This metric can be used only for categorical data, which makes it useful for a table-based partition schema where the data is distributed into cells with categorical feature labels.
\end{definition}

Consider a population $P$ of individuals $i$ distributed in a table-based partitioned schema $S$. The proportion of the population $P$, with a given assignment of feature values corresponding to a bin in the schema $S$, is calculated as $p(bin_{S(i)}) = \frac{|\{i | bin_{S(i)}\}|}{|P|}$. Let the multivariate distribution of the initial feature set be described by $F$ in the schema $S$, and the univariate distribution of the new feature be described by $X$ in the schema $(S+X)$. The observed entropy $H(X)$ is dependent on the distribution of $X$, which remains the same over this operation for any added feature. We assume the size of the population is much larger than the number of bins, i.e., $|P| >> |Range(F + X)|$.
\begin{definition}
\label{def:dispersal}

(Dispersal Ratio) Let the dispersal ratio for a population $P$ with the addition of feature $X$ be defined as: 
$$Disperse(S, X, P) =  |bin_{(S+X)(P)}|  / |bin_{S(P)} |$$
\end{definition}
where $bin_{S(P)}$ is defined as the set of all bins in the histogram corresponding to $P$ distributed in the schema $S$.

Real-world tabular survey data is often comprised primarily of non-ordinal, categorical features (multiple choice answers) that fall neatly into histogram bins. For a sub-population of a fixed size, definition \ref{def:dispersal} ratio captures the increase in the number of histogram bins for the population when a new feature is introduced. Thus, it captures the resulting reduction of the average count per bin and increases to the number of small-count bins.  The lemma below provides some basic intuition behind the choice. The main result of this section is given in Theorem 2.3.  

\begin{lemma}
    An uncertainty coefficient of 1 is equivalent to a dispersal ratio of 1.
    \[U(X|F) = 1 \iff Disperse(S, X, P) = 1\]
\end{lemma}
\begin{lemma}
An uncertainty coefficient of 0 leads to the maximum dispersal ratio.
\[U(X|F) = 0 \implies Disperse(S, X, P) = |Range(X)|\]
\end{lemma}
The above lemmas provide good intuition for our main argument in this section. Let's call these the trivial bounds for dispersal ratio on adding a feature $X$ to the schema. Moreover, we are also interested in quantifying how different patterns of feature correlations, represented by the uncertainty coefficient, impact the dispersal ratio for values within these bounds. The following theorem allows us to establish bounds on dispersal ratio as a function of the uncertainty coefficient or independence.

\begin{theorem} Dispersal Ratio is bounded from above and below as function of the independence of the added feature as follows
\[\frac{|P|\cdot f(u)}{\log(|P|)|Range(F)|}\ge Disperse(S,X,P)\ge \frac{2^{f(u)}}{|Range(F)|}\]
where $f(u) := (1-u)H(X) + H(F)$ with $u = U(X|F)$.
\end{theorem}
\begin{proof}
The following is a short sketch of the proof. The complete proof is detailed in Appendix C.2. 

Let some arbitrary $u = U(X|F)$. Applying the well-known upper bound on entropy \cite{cover1999elements},
\begin{equation}
H(X,F) \le \log_2(|Range(X,F)|)
\end{equation}
provides a lower bound on dispersal ratio. Similarly, the observed entropy can be lower bounded by observing that each record in $S$ discretely contributes to the histogram. This provides an upper bound on the dispersal ratio. 
\end{proof}

Theorem 2.3 shows that, for some fixed value of entropy of the added feature, the non-trivial upper and lower bounds for the dispersal ratio decrease as the uncertainty coefficient increases, and vice-versa, according to the described behaviour of $f(u)$.

Now, we want to compare the effect of adding a new feature $X_1$ or $X_2$ to the schema. Let us assume that $X_1$ is more "independent" than $X_2$ of the distribution of $P$ in $S$ (note that this is equivalent to considering a single feature $X$ and two diverse subpopulations $P1, P2$ with differing relationships to $X$). We can then say that the uncertainty coefficient $u_1 = U(X_1|F)$ is lesser for $X_1$ as compared to $u_2$ corresponding to $X_2$. We can use our results from the above theorem to make this comparison as follows in Theorem 2.4. We define a couple of terms first for ease of notation.

Let the non-trivial lower bound for the dispersal ratio (>1) on adding feature $X$ be denoted as
\begin{equation}
LB(Disperse(S,X,P)) = \frac{2^{(1-u)H(X) + H(F)}}{|Range(F)|}
\end{equation}

Let the non-trivial upper bound for the dispersal ratio (<$|Range(X)|$) on adding feature $X$ be denoted
\begin{equation}
UB(Disperse(S,X,P)) = \frac{|P|\cdot (1-u)H(X) + H(F)}{\log(|P|)|Range(F)|}
\end{equation}

\begin{theorem} Consider two features $X_1$ and $X_2$, identical in terms of entropy, that can be added to the schema. If $X_1$ has higher independence than $X_2$ with respect to $F$, it is equivalent to $X_1$ having a higher LB and higher UB for the dispersal ratio.
\[U(X_1|F) \le U(X_2|F) \iff LB(Disperse(S,X_1,P) \ge LB(Disperse(S,X_2,P) \]
\[U(X_1|F) \le U(X_2|F) \iff UB(Disperse(S,X_1,P) \ge UB(Disperse(S,X_2,P)\]
\end{theorem}
\begin{proof}
This follows from Theorem 2.3. The complete proof is detailed in Appendix C.2.
\end{proof}

\section{Introducing the Diverse Community Excerpts}   %Christine, Gary
\vspace{-1mm}
The Excerpts are in the public domain and were designed to explore algorithm behavior on realistic data with diverse subpopulations (see Figure \ref{fig:dispersal_data}). We selected these data to be \emph{tractable}, in light of a recurring problem identified in the \href{https://www.nist.gov/ctl/pscr/open-innovation-prize-challenges/past-prize-challenges/2018-differential-privacy-synthetic}{NIST Differential Privacy Synthetic Data Challenge} \cite{ridgeway_challenge_2021}, \href{https://www.nist.gov/ctl/pscr/open-innovation-prize-challenges/past-prize-challenges/2020-differential-privacy-temporal}{NIST Differential Privacy Temporal Map Challenge}, and the UNECE High-Level Working Group for the Modernisation of Statistics (HLG-MOS) \href{https://www.nist.gov/ctl/pscr/open-innovation-prize-challenges/past-prize-challenges/2018-differential-privacy-synthetic}{Synthetic Data Test Drive} \cite{unitedsynthetic}, where the target data was too large or complex to identify, diagnose and address shortcomings in the deidendified data. This is a serious problem; consumers of deidentified government data cannot afford to overlook even subtle introductions of bias, artifacts, or privacy leaks. But data properties like subpopulation dispersal can induce defects that are not visible in aggregate utility metrics used by most privacy researchers and data competitions. And deeper exploration is intractable when considering hundreds of features and millions of records. Addressing these issues requires tools designed to make them accessible. 

The Excerpts consist of a curated geography and feature set derived from the significantly larger 2019 American Community Survey (ACS) Public Use Microdata Sample (PUMS) \cite{ACS_2019}, a product of the U.S. Census Bureau. The Census Bureau applies privacy measures to the data\cite{ACS}, and no independent privacy risks of this subset of the data. The Excerpts serve as benchmark data for two currently active, open source projects at the National Institute of Standards and Technology (NIST)--\href{https://doi.org/10.18434/mds2-2943}{SDNist Deidentified Data Report Tool} and the \href{https://pages.nist.gov/privacy_collaborative_research_cycle/}{2023 Collaborative Research Cycle (CRC)}.

The Excerpts' feature set was developed with input from U.S. Census Bureau experts in adaptive sampling design (see \cite{coffey2020}). To identify a small set of communities with challenging, diverse distributions, the Excerpts leverage previous work on geographical differences in CART-modeled synthetic data (see \cite[Appendix B]{Abowd2021}). The open source SDNist library which accompanies the excerpts was developed with input from HLG-MOS participants and synthetic data contractors working with the U.S. Census. In this section we provide an overview of the Excerpts; in the next section we demonstrate their ability to identify and diagnose problems in deidentification algorithms. 

\begin{wrapfigure}{r}{0.5\textwidth}
    \centering
    % \begin{subfigure}[t]{0.49\textwidth}
    %     \centering
  % include second image
        \includegraphics[width = 0.49\textwidth]{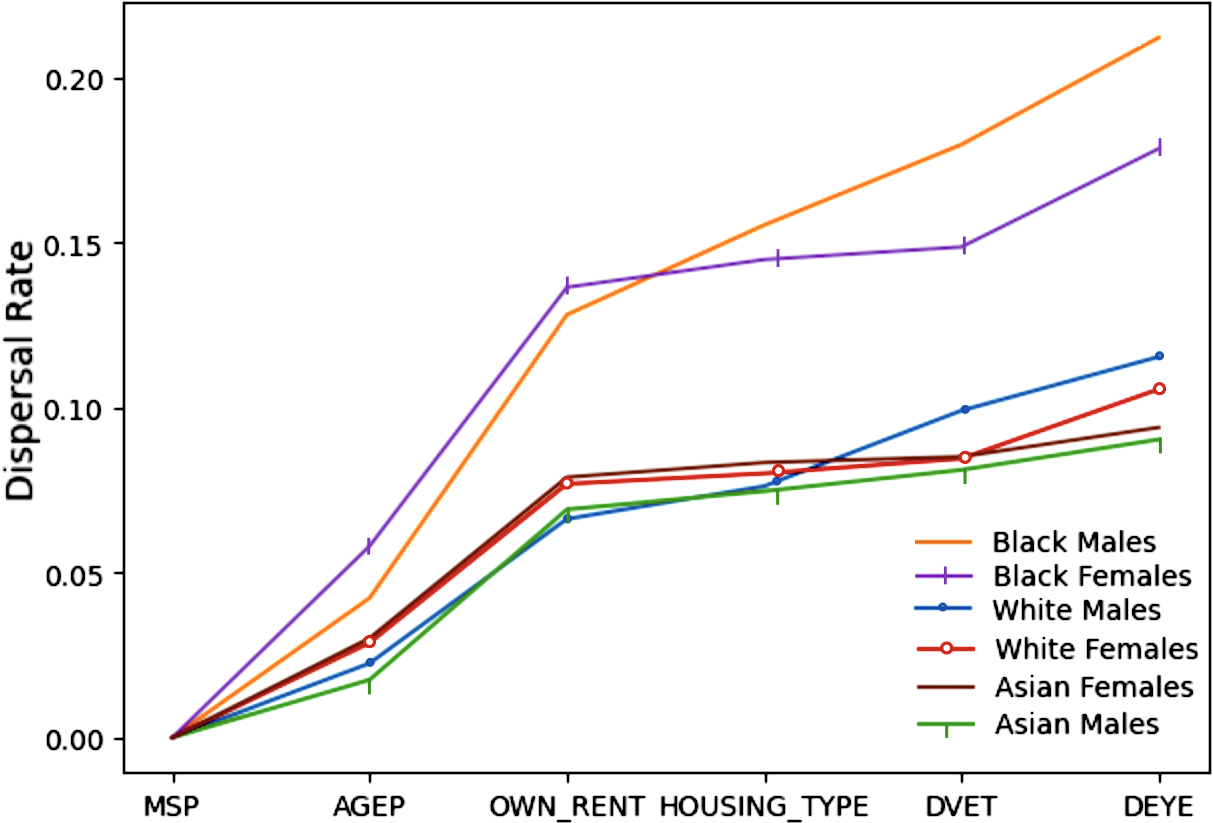} 
        % \subcaption{Effect of feature addition to dispersal rate of the Excerpts}
        % \label{fig:dispersal_data}
    % \end{subfigure}
  %   \begin{subfigure}[t]{0.49\textwidth}
  %       \centering
  %       \includegraphics[width =\linewidth]{figures/marginal_bar_graph.png}  
  %       \subcaption{4-marginal of deidentified instances in the Data Archive by demographic group}
  % %       \label{fig:small_bins}
  %       \end{subfigure}
    \caption{Effect of feature addition to dispersal rate of the Excerpts data (National sample).}
    \label{fig:dispersal_data}
\end{wrapfigure}

% \begin{figure*}
%     \centering
%     \begin{subfigure}[t]{0.5\textwidth}
%     \end{subfigure}
%     \caption{linegraphs}
% \end{figure*}
\vspace{-1mm}

\subsection{Data Overview}
\label{sec:data}
\vspace{-1mm}
\textbf{Feature Selection} The original ACS schema contains over four hundred features, which poses difficulties for diagnosing shortcomings in deidentification algorithms. The Excerpts use a small but representative selection of 24 features, covering major census categories: demographic, household and family, geographic, financial, work and education, disability, and survey weights (discussed in the `Challenges' below). Several Excerpts' features were not in the original ACS features and were designed to provide easy access to certain information. Population \textsc{density} allows models to distinguish rural and urban geographies, \textsc{indp\_cat} aggregates industry codes into categories, \textsc{pincp\_decile} aggregates incomes into percentile bins relative to the record's state, and \textsc{edu} simplifies schooling to focus on milestone grades and degrees. The Excerpts' repository has detailed information including several recommended feature subsets for use in diverse subpopulation analyses. 

\textbf{Recommended Feature Subsets} To support algorithms that require smaller schema and explore the impact of different types and combinations of features, we provide recommended subsets of the 24 features (see Appendix D). Algorithm performance may differ significantly between feature subsets (see Appendix E.4).

\textbf{Subpopulation Dispersal} Figure \ref{fig:dispersal_data} shows subgroup dispersal by race in the Excerpts data; as new features are added that are more independent for one group and more tightly correlated for other groups, the independent group becomes dispersed across the features space. The experiments in Section \ref{sec:selected_algos} use the CRC archive to demonstrate the impact of this on algorithm performance for the largest four race + sex subgroups. However, we note there's nothing `magic' about this choice of demographic subgroups, other population partitions exhibit similar behavior (e.g., partitions by sex + disability, employment sector). A rigorous multidimensional understanding of the relationship between variation in patterns/strength of feature correlations, variations in distribution density, and pitfalls for deidentification is the future research this benchmarking program is designed to support. 

\textbf{Three Excerpts Data Sets}  Data sets of differing sizes, reflecting different types of communities, induce different algorithm behavior. To prevent overfitting research conclusions to a particular type of community, the Excerpts contains three target data sets. The ``Massachusetts'' data contains 7634 records drawn from five Public Use Microdata Areas (PUMAs) of demographically homogeneous communities from the North Shore to the west of the greater Boston area; the ``Texas'' data with 9276 records is drawn from six racially diverse PUMAs of communities surrounding Dallas-Fort Worth, Texas area. The ``National'' data with 27254 records is drawn from 20 diverse PUMAs across the United States; it is an expansion of a challenging PUMA set used during the design of the 2020 Decennial Census Disclosure Avoidance System algorithm \cite{Abowd2021}. The experiments in this paper use the National Excerpts data.

\textbf{Metadata and Documentation:} The Diverse Community Excerpts include \textsc{json} data dictionaries with complete feature definitions, a github readme with detailed \href{https://github.com/usnistgov/SDNist/tree/main/nist%20diverse%20communities%20data%20excerpts}{usage guidance}, and illustrative \href{https://github.com/usnistgov/SDNist/blob/main/nist%20diverse%20communities%20data%20excerpts/national/National%20Excerpts%20Postcard%20Descriptions.pdf}{``postcard''} introductions to the real-world communities in the data. 

%Json files
%Postcards
%Readme

%\subsection{Challenges, Uses and Limitations}
%[This data was designed to stress-test approaches and highlight their behaviors in real-world contexts. Previously we've discussed the impact of diverse populations; in this section we review other challenges and opportunities available in this feature set. We also discuss some limitations, things we'd like to check but are not inluded in this version yet]

\textbf{Challenges:} The Excerpts have been designed to be representative of real-world survey data conditions. In addition to diversity, these include other challenging properties:  logical constraints between features (e.g., \textsc{age} = 6 places a constraint on \textsc{marital status}) that can be difficult for synthetic data generators (see Figure \ref{fig:propensity}). Additional modeling challenges include heterogeneous feature types and cardinalities (e.g., \textsc{DEYE} has two categorical values while \textsc{industry} has 271, \textsc{income} is an integer while \textsc{poverty  income ratio} is continuous). Uneven feature granularity can amplify problems with unequal subpopulation dispersal (ACS 2019 \textsc{rac1p} uses one code for Asian, but four detailed codes for Native American). We also include the sampling weights that survey data users need to simulate a full population; these lose their original meaning after deidentification changes the data sample, and this is a largely unsolved problem.

\textbf{Limitations:} Although ACS data consumers generally assume the data remains representative of the real population, the ACS PUMS data has had basic statistical disclosure control deidentification applied (as noted above), which may impact its distribution. Additionally, there are shortcomings in the Diverse Community Excerpts that we plan to address in future versions: the Excerpts do not currently include the ACS Household IDs (which would support joining individuals in the same households for social network synthesis), Individual IDs (relevant for reidentification research), or a clear training/testing partition (important for differential privacy research). 

\vspace{-1mm}
% ------------------------------------------------------
%  SDNist 
% ------------------------------------------------------
\section{Introducing the SDNist Deidentified Data Report Tool} 
\label{sec:sdnist}

The \href{https://doi.org/10.18434/mds2-2943}{SDNist Deidentified Data Report Tool} is a Python library that is undergoing active development. SDNist evaluates fidelity, utility, and privacy of a given set of deidentified instances of the Excerpts and generates human- and machine-readable summary quality reports enumerated and illustrated for each utility and privacy metric. It is comprised largely of fidelity metrics drawn from data stakeholders through the \href{https://pages.nist.gov/HLG-MOS_Synthetic_Data_Test_Drive/}{HLG-MOS Synthetic Data Test Drive}, and aims to provide a comprehensive view of algorithm behavior. Examples of complete reports, including detailed metric definitions, are linked in the Appendix E. A metric overview is below; usage is demonstrated in Section \ref{sec:demo}. 

\textbf{Univariate and Correlation Metrics} The SDNist report covers univariate feature distributions as well as pairwise feature correlations with Pearsons and Kendall Tau Coefficients. 

\textbf{Higher Order Similarity Metrics ($k$-marginal, equivalent subsample)} Data analysts often use more than two features. SDNist provides higher order distributional similarity comparison using the \emph{$k$-marginal} edit distance similarity score  \cite{ridgeway_challenge_2021, Abowd2021, lothe_sdnist_2021}. This compares the target data and deidentified data densities across $k$-marginal queries, with a maximum score of 1000 indicating identical data sets. The report provides overall and and per-PUMA 3-marginal scores. To aid interpretation, the report includes an \emph{Equivalent Subsample (ES)} table that provides a comparison to edit distance induced by sampling error. A $k$-marginal score with an ES of 15\% is about as dissimilar from the target data as a data set created by uniformly randomly discarding 85\% of the target data. 

\textbf{Task Comparison Metrics (propensity, linear regression)} The report includes two machine learning task-based metrics. The \emph{propensity} metric (Figure \ref{fig:propensity}) trains a classifier to predict whether a record belongs to the target or deidentified data. The classifier's per-record prediction confidence (propensity distribution) is plotted for both data sets; when the classifier finds it impossible to distinguish the two, the traces will match with a spike at 50 \%. The \emph{linear regression} metric focuses on a task very relevant to the subpopulation dispersal problem (see Figure 6). It uses linear regression to summarize the relationship between educational attainment and income decile for race + sex subgroups. The target regression line is given in red, the deidentified line is in green. Correlation strength differs for different subgroups; this causes dispersal and introduces challenges for deidentification. Some deidentification approaches artificially weaken the correlation, and others strengthen it. 

\textbf{Visual Diagnosis Metrics (PCA, regression heatmap)} To help diagnose the causes of low similarity scores, and identify bias/artifacts introduced during deidentification, the SDNist report includes several visualizations. \emph{Pairwise Principle Component Analysis (PCA)} supports direct comparison of the target and deidentified data distribution using scatterplots across principle component axes. An interactive PCA exploration tool is \href{https://pages.nist.gov/privacy_collaborative_research_cycle/pages/archive.html#pca-tool}{available}. In Figure \ref{fig:propensity}, the target data (ground truth) PCA scatterplot is shown on the left in blue, the deidentified data plot is shown on the right, in green. For additional insight, records with marital status = N/A (indicating  children < 15) are highlighted in red in both plots. Algorithms that provide good privacy and utility will reproduce the shape of the target scatterplot using new points (i.e., deidentified records). 
The \emph{linear regression heatmaps} (Figure \ref{fig:regression}) provide another resource for identifying artifacts. The target data distribution is shown as a red-blue heatmap of distribution densities (normalized by educational attainment level). The deidentified heatmaps show deviation from the target distribution: brown indicates the deidentified data contains too many individuals in that category, purple indicates it contains too few. Artifacts and bias are identifiable as large ares of blue or brown blocks, indicating where privacy techniques systematically erase or over emphasize parts of the distribution. 

\textbf{Empirical Privacy Metrics (Unique Exact Match)}  Privacy protection can be considered empirically, or formally. For techniques with parameterized guarantees (such as differential privacy or $k$-anonyimity) we include parameter values in their metadata. However, any technique (including DP) may have catastrophic privacy failures.  The \emph{Unique Exact Match (UEM)} metric is a measure of unambiguous privacy leakage applicable to all deidentification techniques; it simply counts the percentage of unique individuals in the target data who remain present unaltered in the deidentified data. Being unique, these individuals are more vulnerable to reidentification; being unaltered, this reidentification may not be difficult. As Table \ref{deid_table} (and Appendix E.3) shows, a simple differentially private hisotgram technique with weak guarantee $\epsilon = 10$ can reproduce its input data nearly exactly without violating DP. The UEM metric allows us to identify when one technique is performing very badly on privacy, or identify when one technique is performing significantly worse than another. Because UEM doesn't test non-trivial reidentification attacks, it cannot be used to determine whether a technique provides good privacy, or which of two similarly performing techniques is best. We plan to expand the SDNist privacy metrics; however, this starting point provides valuable insights. 

\vspace{-1mm}

\section{The Collaborative Research Cycle Data and Metrics Archive} 
\label{sec:demo}

\vspace{-1mm}

The \href{https://github.com/usnistgov/privacy_collaborative_research_cycle/ }{CRC Data and Metrics Archive} is comprised of deidentified data samples generated by applying deidentification algorithms to the Diverse Community Excerpts, using varied techniques, feature subsets, and parameters.  It opened collection in March 2023 and is still actively accepting submissions; at time of writing it includes 453 samples of 33 techniques from 18 deidentification libraries.  Below we provide a brief overview of its contents, based on meta-data available in the archive \href{https://github.com/usnistgov/privacy_collaborative_research_cycle/blob/research-acceleration-bundle/crc_data_and_metric_bundle_1.1/index.csv}{index file}.

\textbf{Privacy Types.} We delineate three basic types:  56 samples use Statistical Disclosure Control Techniques that use perturbation or redaction to anonymize the target data while leaving it substantially intact; 139 are non-differentially private synthetic data techniques which generate new records to fit the target distribution; and 258 are Differentially Private techniques which add randomization to satisfy a formally private guarantee limiting an individual record's influence on the relative probability of possible outputs. \cite{Dwork2006}

\textbf{Algorithm Types.} To support meta-analysis, we provide a high-level categorization of approaches.  All approaches can be implemented with or without formal privacy guarantees.  In the archive, 166 deidentified samples use neural network modeling approaches, 143 fit traditional statistical models, 67 iteratively update a naive distribution to mimic query results on the target data, 56 use statistical disclosure control anonymization, 12 are simple histogram techniques that alter counts of record occurrences, and nine select new points using geometric interpolation between target records.

\textbf{Exploring Diverse Subgroup Dispersal in the Archive.}  In Section \ref{sec:math} and Figure \ref{fig:orange_blue} we discussed the impact subpopulation dispersal and resulting small count cells can have on performance after deidentification. In the analyses below we look at how this impact plays out in practice in the archive. These deidentified samples used the National Excerpts target data and demographic-focused feature set.  To compare subgroup utility, we use the $k$-marginal metric across 4-marginals including the race + sex features. Only samples with 15\% subsample equivalence  or greater are included. 

\begin{wrapfigure}{r}{0.5\textwidth}
    \centering
     \includegraphics[width =0.49\textwidth]{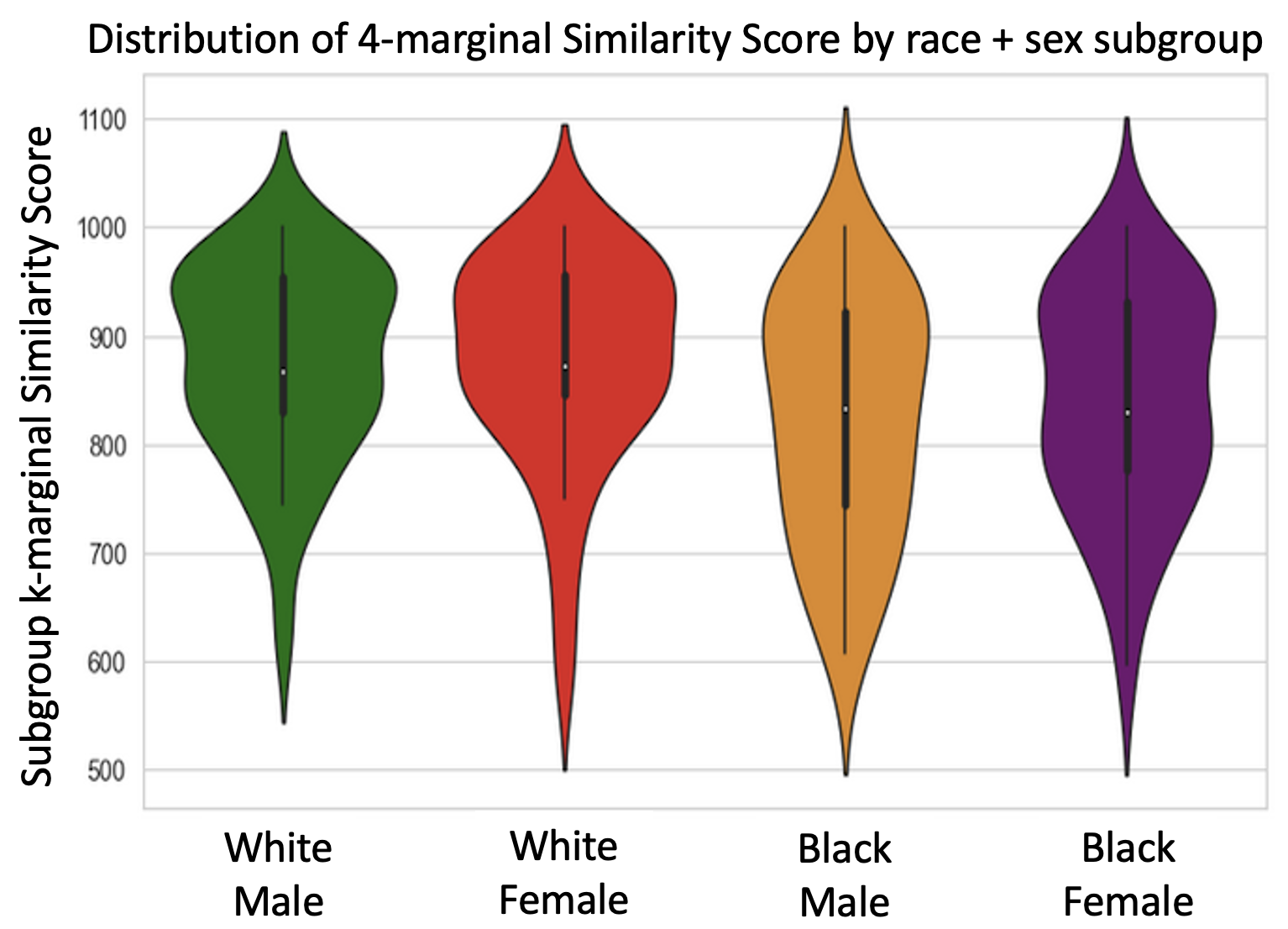} 
     \caption{4-marginal edit distance similarity, by demographic subgroup, from the CRC Data and Metrics Archive}
     \label{fig:race_marginals}
\end{wrapfigure}

Figure \ref{fig:race_marginals} shows the distribution of $k$-marginal scores for the four largest race + sex subpopulations; the more dispersed group in general has poorer performance, but some deidentified data samples show good performance overall.  Understanding what distinguishes deidentification algorithms that are more or less impacted by subgroup dispersal is a key goal of this benchmarking  project, and identifying them is a first step.  Figure \ref{fig:utility-marginal} provides further breakdown, showing how subgroup performance varies with overall data fidelity.  We see that the expected difference in performance appears regularly but not constantly, some algorithms overcome it and provide similar performance for all groups.  Figure \ref{fig:utility-UEM} supplements this with privacy leak information using the UME metric (Section \ref{sec:sdnist}). Some high performing algorithms are actually identically reproducing the target data, providing very little privacy, while others are able to maintain distributional fidelity for all subgroups using substantially new records.  The CRC benchmarking program is designed to support the research that leads to a formal understanding of this behavior, and the development of high fidelity, equitable deidentification techniques with negligible privacy leaks. The detailed metrics for all samples appearing in these charts are provided in Appendix E2; no current archive algorithm fully satisfies all three requirements for this feature set. 

\begin{figure*}[h]
    \centering
    \begin{subfigure}[t]{0.49\textwidth}
    \includegraphics[width = \linewidth]{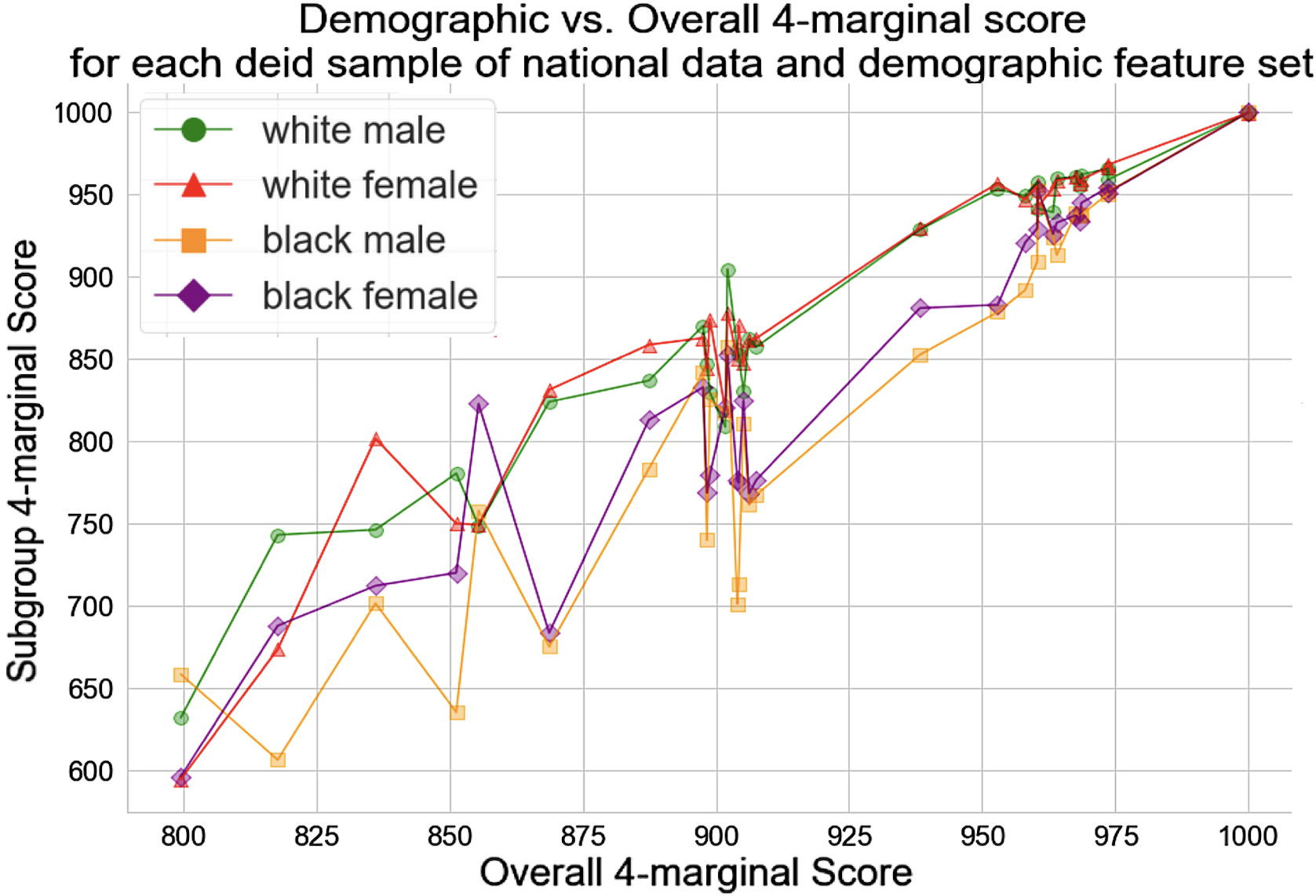}
    \subcaption[]{$k$-marginal similarity performance.}
    \label{fig:utility-marginal}
    \end{subfigure}
    \begin{subfigure}[t]{0.49\textwidth}
    \includegraphics[width = \linewidth]{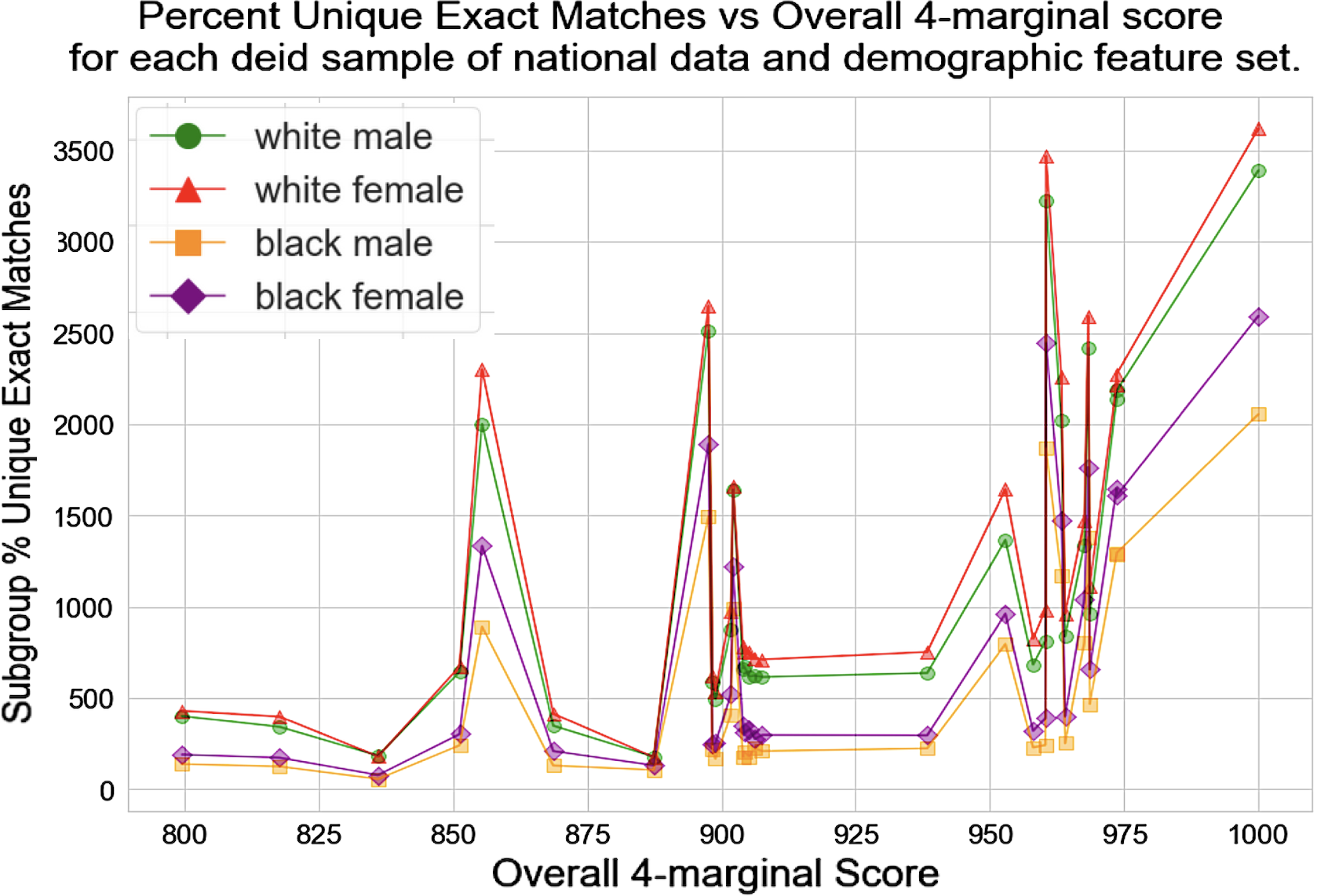}
    \subcaption[UEM performance]{UEM privacy performance.}
    \label{fig:utility-UEM}
    \end{subfigure}
    \caption{{Subgroups utility/privacy performance, ordered by total population 4-marginal performance.}}
    \label{fig:utility-graphs}
\end{figure*}
% \vspace{-1mm}
\subsection{Exploration of Selected Algorithms}
\label{sec:selected_algos}

%\textcolor{red}{Say now that you have the toolkit you're able to seriously start to document and catalogue these bad behaviors, and you've gone from knowing something's probably up, to documenting how it's up, and future work will be more math on the open questions and how to address it.}

%two publicly available resources: metrics from \href{SDNist}{ https://github.com/usnistgov/SDNist} \cite{task_nist_2022}, and deidentified data samples from the CRC Data and Metrics archive \cite{task_sdnist_2023}.  The archive currently contains over 300 deidentification samples based on the Diverse Community Excerpts target data; we've selected

Table \ref{deid_table} uses seven interesting Archive contributions to demonstrate the Excerpts' efficacy as a tool for exploring and understanding the behavior of deidentification algorithms.

\begin{table}[h]
\centering
\resizebox{\textwidth}{!}{%
\begin{tabular}{l r r r r}
% \hline
\\[-1em]
\multirow{2}{*}{} \textbf{Library and Algorithm } & \textbf{Privacy Type} & \textbf{Algorithm Type} & \textbf{Priv. Leak (UEM)} & \textbf{Utility (ES)} \\
\hline
\\[-1em]
DP Histogram ($\epsilon = 10$) & differential privacy (DP) & simple histogram & 100\% & $\sim 90$\% (988)\\
\hline
\\[-1em]
\multirow{2}{*}{}R synthpop CART model \cite{synthpop} & non-DP synthetic data & \shortstack{multiple imputation\\ decision tree} & 2.54\% & $\sim$40\% (935)\\
\hline
\\[-1em]
\multirow{2}{*}{}MOSTLY AI SDG \cite{MostlyAI} \cite{platzer2022rule} & non-DP synthetic data & \shortstack{proprietary pre-trained\\ neural network} & 0.03\% & $\sim$30\% (921)\\
\hline
\\[-1em]
\multirow{2}{*}{}SmartNoise MST ($\epsilon = 10$) \cite{mst} & DP & \shortstack{probabilistic graphical\\ model (PGM)} & 13.6\% & $=10$\% (969)\\
\hline
\\[-1em]
\multirow{2}{*}{}SDV CTGAN \cite{SDV} \cite{ctgan} & non-DP synthetic data & \shortstack{generative adversarial\\ network (GAN)} & 0.0\% & $\sim$5\% (775)\\
\hline
\\[-1em]
SmartNoise PACSynth ($\epsilon = 10$) \cite{pacsynth} & DP + $k$-anonymity & constraint satisfaction & 0.87\% & $\sim$1\% (551)\\
\hline
\\[-1em]
synthcity ADSGAN \cite{ads} \cite{synthcity}& custom noise injection & GAN & 0.0\% & $<1$\% (121)\\
\hline
\hspace{1mm}
\end{tabular}}
\label{tab:summary2}
\caption{ \textbf{Performance of selected deidentification algorithms,} \small see Appendix E.1 for additional details. }
\label{deid_table}
\vspace{-1mm}
\end{table}

\begin{figure}[h]
     \centering
     \includegraphics[width=\textwidth]{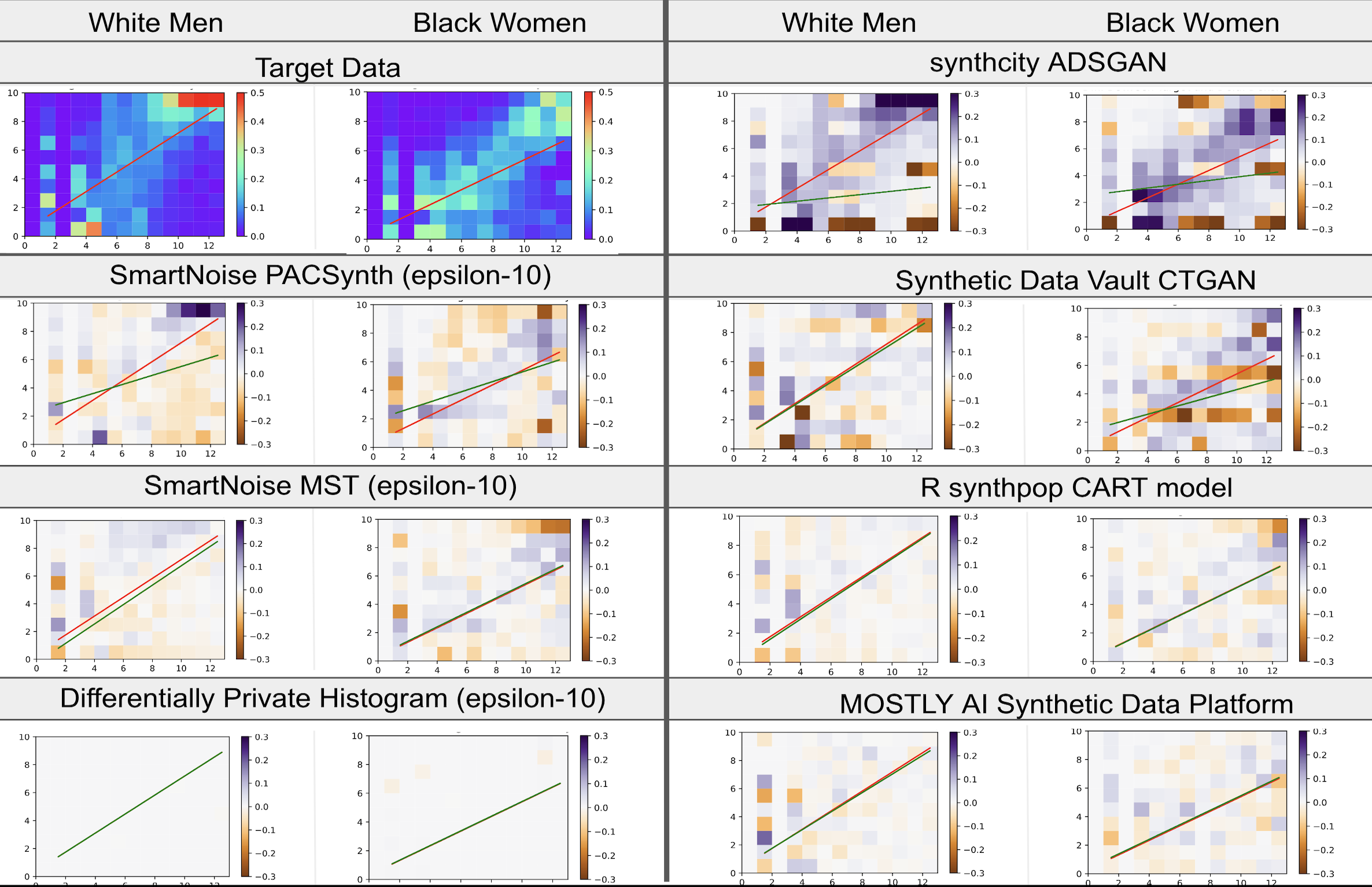}
     \caption{Linear regression metric showing how well the selected algorithms maintain the relationship between educational attainment ($x$). and income decile ($y$) for different demographic groups.}
     \label{fig:regression}
 \end{figure}

The regression metric (Figure \ref{fig:regression}) allows us to explore the impact of variable subpopulation dispersal. Very poor performing algorithms show little impact (ADSGAN, PACSynth), while others have worse performance on the more dispersed subpopulation (black women).  CART and MST perform well, but introduce a slight bias that strengthens the correlation between high educational attainment and high income in the dispersed group (brown artifact in upper right of heatmap). This could obscure the real disparity between the groups.  

Considering privacy, all algorithms in the left column guarantee differential privacy with an identical privacy parameter setting $\epsilon$=10 (weak privacy protection). However we see in Table \ref{deid_table} that techniques with the same formal guarantee can produce widely varied results in terms of both privacy and utility.  PACSynth has additional $k$-anonymity protection which eliminates rare, dispersed records; this provides good privacy, but has much poorer utility with unusual impacts on the distribution that can be identified in this and other metrics.  Meanwhile, the simple DP Histogram provides almost no privacy protection at all, reproducing the target data nearly exactly.  MST is possibly good compromise of privacy and utility (and would provide better privacy at smaller $\epsilon$), but non-differentially private techniques CART and MostlyAI have significantly less trivial privacy leakage (by UEM metric, Section \ref{sec:sdnist}) and better utility. 

\begin{figure}[h]
    \centering
     \includegraphics[width=\textwidth]{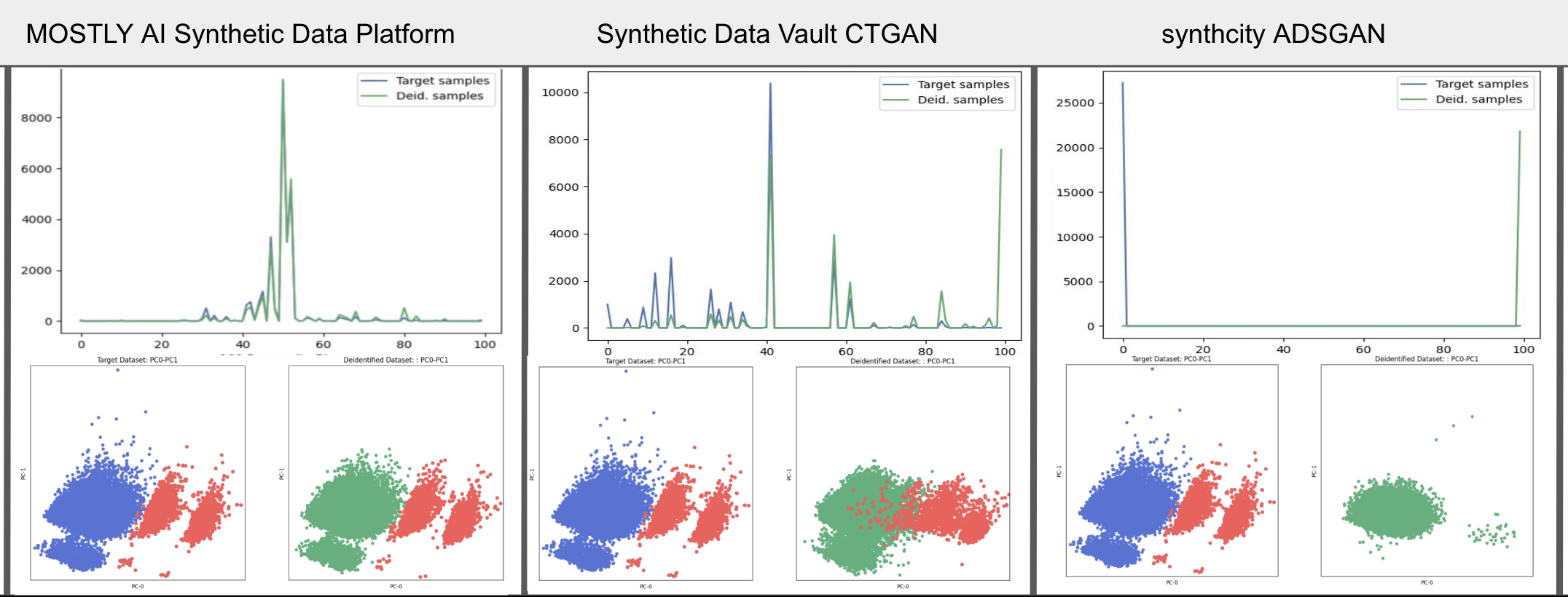}  
  \caption{Propensity and PCA metrics. }
  \label{fig:propensity}
  \end{figure}

 In figure \ref{fig:propensity} we present two more SDNist metrics: The PCA metric uses scatterplots along principle component axes to compare the shapes of distributions. The propensity metric trains a classifier to distinguish between real and synthetic data; if the data are indistinguishable the propensity distributions will peak at the center (indicating the classifier can only make a `50/50 guess' whether a given record is from the target data or deidentified data).  If the deidentified data contains artifacts or bias these will be visible as mismatches in the shape of the scatterplot.  This figure compares the three neural network synthesizers in our selection, and provides further explanation for the utility values noted in Table \ref{deid_table}. CTGAN is not preserving constraints that hold for child records (highlighted in red), and is missing a constraint on housing features (visible as the three well-separated clusters in the target data, but blurred together in the CTGAN data). Examples of records violating constraints might be a 7-year-old widow or an inmate who owns her jail; these can be confidently classified as synthetic by the propensity metric.  Meanwhile, the ADSGAN data contains no children at all, and has retained none of the structure of the target data. Note that these results are more informative than simple edit distance utility scores.  We are able to identify and diagnose specific algorithm behaviors induced by the real world challenges in the Excerpt data.

%%%%%%%%%%%%%
%%%%%%%%%%%%%%%%%%%%%%%%%%%%%%%%%%%%%%%%%%%%

  %Gary

% \textcolor{red}{expand with a table summarizing the offerings of these guys, and get a line in about how us nailing things down to a specific data set is invaluable because it allows the sort of comparisions we've got elsewhere in the paper. Mostly, put more words in here to make it clear we care about their feedback-- let's save at least 2/3 of a page for this} 
% Other things Exist
%----There's survey papers (tumult, clare, gillian, others?)
%----There's evaluation platforms (SDGym, etc, ) 
\begin{table}[h]
    \scriptsize
    \centering
    \begin{tabular}{l c c c c c c}
         \shortstack{\textbf{Tool} \\ \textbf{name}} &
         \shortstack{\textbf{Univariate /} \\ \textbf{Correlations}} & 
         \shortstack{ \textbf{Higher Order} \\ \textbf{Simililarity}} &
         \shortstack{\textbf{Task} \\ \textbf{Comparison}} &
         \shortstack{\textbf{Visual} \\ \textbf{Diagnosis}} & \shortstack{\textbf{Privacy}\\\textbf{analysis}} &
         \shortstack{\textbf{Built-in} \\
         \textbf{synthesizer(s)}} \\
         \hline

        \href{https://github.com/statice/anonymeter}{Anonymeter}\cite{anonymeter} & & & & & \checkmark \\
        
        \href{https://github.com/DataResponsibly/DataSynthesizer}{Data Responsibly}\cite{ping_datasynthesizer_2017} &
        \checkmark & \checkmark  &           &           &  & \checkmark \\
        \href{https://nhsx.github.io/skunkworks/synthetic-data-pipeline}{NHS Synthetic Data Pipeline} & \checkmark & \checkmark & \checkmark & \checkmark & \checkmark & \checkmark* \\
        \href{https://github.com/opendp/opendp}{OpenDP}\cite{The_OpenDP_Team_OpenDP_Library} &
        \checkmark &            &           &           &  & \checkmark*  \\
        \href{https://www.synthpop.org.uk}{R Synthpop}\cite{synthpop} &
         \checkmark & \checkmark & \checkmark &  & \checkmark & \checkmark\\
         \href{https://github.com/sdv-dev/SDGym}{SDGym / SDMetrics}\cite{SDV} &
         \checkmark & \checkmark & \checkmark & &\checkmark & \checkmark  \\
         % \href{https://github.com/sdv-dev/SDMetrics}{SDMetrics}\cite{SDV}  &
         % \checkmark & & & & \checkmark  \\
         \href{https://github.com/usnistgov/SDNist}{\textbf{SDNist}} & 
         \checkmark & \checkmark & \checkmark &  \checkmark &
         \checkmark \\
         \href{https://github.com/vanderschaarlab/synthcity}{Synthcity}\cite{synthcity} &
         \checkmark & \checkmark & \checkmark &            & \checkmark & \checkmark \\
         \href{https://github.com/Baukebrenninkmeijer/table-evaluator}{Table Evaluator} & 
         \checkmark & \checkmark & 
         \checkmark & \checkmark \\

         \href{https://github.com/alan-turing-institute/privacy-sdg-toolbox}{TAPAS} &
         & & & & \checkmark \\
         
         \href{https://github.com/ydataai}{YData} &
         \checkmark & \checkmark & \checkmark & \checkmark &  & \checkmark \\
        \hline
    \end{tabular}
    \caption{Synthetic data evaluation tools}
    \label{tab:synth_data_tools}
    % \hspace{-1cm}
    \vspace{-5mm}
\end{table}
\vspace{-2mm}
\section{Related Works}
% \vspace{-1mm}

Table \ref{tab:synth_data_tools} provides a broad survey of comparable tools for benchmarking and evaluating synthetic data, including both library links and research citations. Libraries annotated with `*' are specifically configured to evaluate data from their built-in synthesizer models. Visual diagnosis refers to tools that visualize  and explore data within a high dimensional space (e.g., principal component analysis). Privacy analysis specifically refers to comparison of output records to input records (e.g., re-identification, replication). Many of these tools have built in data and all of them also have `bring-your-own-data' capabilities. SDNist, and the CRC program in general, distinguish themselves with a more extensive evaluation set, and a specific goal to promote the state of research in the deidentification field as a whole, aiming at rigorously well understood high performance on diverse, real-world data. Unlike the CRC, none of these resources are designed to focus energy on improving synthetic data specifically by comparing algorithm performance on unified data with common metrics. 

\section{Conclusion}

Rather than focus on evaluating arbitrary data, we present tools that are tailored to a specific dataset. 
We are motivated by previous evaluations of differentially private \cite{bowen_comparative_2021} and of GAN-based synthetic data \cite{tao_benchmarking_2022} generators that show surprising differences in algorithm performance. The Excerpts, partnered with SDNist and theory presented here, provide a compact vehicle to address diverse subpopulations, which we show is a persistent problem facing deidentification technologies.

\textbf{Disclaimer.} The Collaborative Research Cycle (CRC), the Data and Metrics Archive, and SDNist are intended as tools to encourage investigation and discussion of deidentification algorithms, and they are not intended or suitable for product evaluation. The National Institute of Standards and Technology does not endorse any algorithm included in these resources. No mention of a commercial product in this paper or any CRC resource constitutes an endorsement. 

\small
\bibliographystyle{unsrtnat}
\bibliography{references}

\end{document}

\maketitle

\renewcommand{\contentsname}{Supplemental Information Contents}
\setcounter{tocdepth}{2} 
\tableofcontents

\appendix

\section{Disclaimer}

The Collaborative Research Cycle (CRC), the Data and Metrics Archive, and SDNist are intended as tools to encourage investigation and discussion of deidentification algorithms, and they are not intended or suitable for product evaluation. The National Institute of Standards and Technology does not endorse any algorithm included in these resources. No mention of a commercial product in this paper or any CRC resource constitutes an endorsement.

\hypertarget{dataset-provisions}{%
\section{Dataset Provisions}\label{dataset-provisions}}

\hypertarget{dataset-url}{%
\subsection{Dataset URLs}\label{dataset-url}}

\begin{itemize}

    \item \href{https://data.nist.gov/od/id/mds2-2895}{The NIST public data
repository access point.}
    \item \href{https://github.com/usnistgov/SDNist/tree/main/nist\%20diverse\%20communities\%20data\%20excerpts}{Direct
data access link.}
    \item \href{https://github.com/usnistgov/SDNist/blob/main/nist\%20diverse\%20communities\%20data\%20excerpts/data_dictionary.json}{Direct
data access to the data dictionary.} 
  \item The dataset DOI is \href{https://doi.org/10.18434/mds2-2895}{10.18434/mds2-2895.}
\end{itemize}

\hypertarget{dataset-format-notes}{%
\subsubsection{Dataset Format Notes}\label{dataset-format-notes}}

The raw data are in comma-sperated value (CSV) format with (Javascript object notation) JSON data dictionaries defining valid
values.

The NIST data repository has a structured metadata retrieval system that
interfaces with \href{https://www.data.gov}{data.gov} and conforms to
\href{https://www.go-fair.org/fair-principles/}{FAIR principles} and the
\href{https://strategy.data.gov/}{best practice for Federal Data
Strategy}. See additional information
\href{https://data.nist.gov/sdp/\#/about}{here}.

\hypertarget{author-statement}{%
\subsection{Author Statement}\label{author-statement}}

The authors bear all responsibility in case of violation of rights. We
have confirmed licensing and provide detailed information in the article
and in the dataset datasheet.

\hypertarget{hosting-and-licensing}{%
\subsection{Hosting and Licensing}\label{hosting-and-licensing}}

The data associated with this publication were created, hosted, and
maintained, by the National Institute of Standards and Technology in
their \href{https://data.nist.gov/od/id/mds2-2895}{permanent data
repository}, in perpetuity.

\noindent The data are in the public domain. \href{https://www.nist.gov/open/copyright-fair-use-and-licensing-statements-srd-data-software-and-technical-series-publications}{NIST statement on software and data}:

\noindent "NIST-developed software is provided by NIST as a public service. You may use, copy, and distribute copies of the software in any medium, provided that you keep intact this entire notice. You may improve, modify, and create derivative works of the software or any portion of the software, and you may copy and distribute such modifications or works. Modified works should carry a notice stating that you changed the software and should note the date and nature of any such change. Please explicitly acknowledge the National Institute of Standards and Technology as the source of the software. 
NIST-developed software is expressly provided "AS IS." NIST MAKES NO WARRANTY OF ANY KIND, EXPRESS, IMPLIED, IN FACT, OR ARISING BY OPERATION OF LAW, INCLUDING, WITHOUT LIMITATION, THE IMPLIED WARRANTY OF MERCHANTABILITY, FITNESS FOR A PARTICULAR PURPOSE, NON-INFRINGEMENT, AND DATA ACCURACY. NIST NEITHER REPRESENTS NOR WARRANTS THAT THE OPERATION OF THE SOFTWARE WILL BE UNINTERRUPTED OR ERROR-FREE, OR THAT ANY DEFECTS WILL BE CORRECTED. NIST DOES NOT WARRANT OR MAKE ANY REPRESENTATIONS REGARDING THE USE OF THE SOFTWARE OR THE RESULTS THEREOF, INCLUDING BUT NOT LIMITED TO THE CORRECTNESS, ACCURACY, RELIABILITY, OR USEFULNESS OF THE SOFTWARE.
You are solely responsible for determining the appropriateness of using and distributing the software and you assume all risks associated with its use, including but not limited to the risks and costs of program errors, compliance with applicable laws, damage to or loss of data, programs or equipment, and the unavailability or interruption of operation. This software is not intended to be used in any situation where a failure could cause risk of injury or damage to property. The software developed by NIST employees is not subject to copyright protection within the United States."

\hypertarget{datasheet-for-dataset-nist-diverse-communities-data-excerpts}{%
\section{Datasheet for dataset ``NIST Diverse Communities Data
Excerpts''}\label{datasheet-for-dataset-nist-diverse-communities-data-excerpts}}

Questions from the \href{https://arxiv.org/abs/1803.09010}{Datasheets
for Datasets} paper, v7.

Jump to section:

\begin{itemize}
% \tightlist
\item
  \protect\hyperlink{motivation}{Motivation}
\item
  \protect\hyperlink{composition}{Composition}
\item
  \protect\hyperlink{collection-process}{Collection process}
\item
  \protect\hyperlink{preprocessingcleaninglabeling}{Preprocessing/cleaning/labeling}
\item
  \protect\hyperlink{uses}{Uses}
\item
  \protect\hyperlink{distribution}{Distribution}
\item
  \protect\hyperlink{maintenance}{Maintenance}
\end{itemize}

\hypertarget{motivation}{%
\subsection{Motivation}\label{motivation}}

\hypertarget{for-what-purpose-was-the-dataset-created}{%
\subsubsection{For what purpose was the dataset
created?}\label{for-what-purpose-was-the-dataset-created}}

The NIST Diverse Communities Data Excerpts (the Excerpts) are
demographic data created as benchmark data for deidentification
technologies.

The Excerpts are designed to contain sufficient complexity to be
challenging to de-identify and with a compact feature set to make them
tractable for analysis. We also demonstrate the data contain
subpopulations with varying levels of feature independence, which leads
to small cell counts, a particularly challenging deidentification
problem.

The Excerpts serve as benchmark data for two open source projects at the
National Institute of Standards and Technology (NIST): the
\href{https://doi.org/10.18434/mds2-2943}{SDNist Deidentified Data
Report tool} and the
\href{https://pages.nist.gov/privacy_collaborative_research_cycle/}{2023
Collaborative Research Cycle} (CRC).

\hypertarget{who-created-the-dataset-e.g.-which-team-research-group-and-on-behalf-of-which-entity-e.g.-company-institution-organization}{%
\subsubsection{Who created the dataset (e.g., which team, research
group) and on behalf of which entity (e.g., company, institution,
organization)?}\label{who-created-the-dataset-e.g.-which-team-research-group-and-on-behalf-of-which-entity-e.g.-company-institution-organization}}

The Excerpts were created by the
\href{https://www.nist.gov/itl/applied-cybersecurity/privacy-engineering}{Privacy
Engineering Program} of the \href{https://www.nist.gov/itl}{Information
Technology Laboratory} at the \href{https://www.nist.gov}{National
Institute of Standards and Technology} (NIST).

The underlying data was published by the U.S. Census Bureau as part of
the 2019 American Community Survey (ACS)
\href{https://www.census.gov/programs-surveys/acs/microdata/access.2019.html\#list-tab-735824205}{Public
Use Microdata Sample} (PUMS).

\hypertarget{who-funded-the-creation-of-the-dataset}{%
\subsubsection{Who funded the creation of the
dataset?}\label{who-funded-the-creation-of-the-dataset}}

The data were collected by the U.S. Census Bureau, and the Excerpts were
curated by NIST. Both are U.S. Government agencies within the Department
of Commerce. Aspects of the Excerpts creation were supported under NIST
contract 1333ND18DNB630011.

\hypertarget{any-other-comments}{%
\subsubsection{Any other comments?}\label{any-other-comments}}

No.

\hypertarget{composition}{%
\subsection{Composition}\label{composition}}

\hypertarget{what-do-the-instances-that-comprise-the-dataset-represent-e.g.-documents-photos-people-countries}{%
\subsubsection{What do the instances that comprise the dataset represent
(e.g., documents, photos, people,
countries)?}\label{what-do-the-instances-that-comprise-the-dataset-represent-e.g.-documents-photos-people-countries}}

The instances in the data represent individual people.

The Excerpts consist of a small curated geography and feature set
derived from the significantly larger 2019 American Community Survey
(ACS) Public Use Microdata Sample (PUMS), a publicly available product
of the U.S. Census Bureau. The original ACS schema contains over four
hundred features, which poses difficulties for accurately diagnosing
shortcomings in deidentification algorithms. The Excerpts use a small
but representative selection of 24 features, covering major census
categories: Demographic, Household and Family, Geographic, Financial,
Work and Education, Disability, and Survey Weights. Several Excerpts
features are derivatives of the original ACS features, designed to
provide easier access to certain information (such as income decile or
population density).

There is only one type of instance. All records in the data represent
separate, individual people.

\hypertarget{how-many-instances-are-there-in-total-of-each-type-if-appropriate}{%
\subsubsection{How many instances are there in total (of each type, if
appropriate)?}\label{how-many-instances-are-there-in-total-of-each-type-if-appropriate}}

There are three geographic partitions in the data. See the ``postcards''
and data dictionaries in each respective directory for more detailed
information. Instances in partitions: 
\begin{itemize}
\item \texttt{national}: 27254 records
\item \texttt{massachusetts}: 7634 records 
\item \texttt{texas}: 9276 records
\end{itemize}

\hypertarget{does-the-dataset-contain-all-possible-instances-or-is-it-a-sample-not-necessarily-random-of-instances-from-a-larger-set}{%
\subsubsection{Does the dataset contain all possible instances or is it
a sample (not necessarily random) of instances from a larger
set?}\label{does-the-dataset-contain-all-possible-instances-or-is-it-a-sample-not-necessarily-random-of-instances-from-a-larger-set}}

The data set is a curated sample of the ACS by geography, with a reduced
feature set designed to provide a tractable foundation for benchmarking
deidentification algorithms (24 features rather than the original ACS's
400 features). Geographically it is comprised of 31 Public Use Microdata
Areas 
\begin{itemize}
\item \texttt{national}: 27254 records drawn from 20 Public Use
Microdata Areas (PUMAs) from across the United States. This excerpt was
selected to include communities with very diverse subpopulation
distributions. 
\item \texttt{texas}: 9276 records drawn from six PUMAs of
communities surrounding Dallas-Fort Worth, Texas area. This excerpt was
selected to focus on areas with moderate diversity. 
\item \texttt{massachusetts}: 7634 records drawn from five PUMAs of
communities from the North Shore to the west of the greater Boston,
Massachusetts area. This excerpt was selected to focus on areas with
less diversity.
\end{itemize}

\hypertarget{what-data-does-each-instance-consist-of}{%
\subsubsection{What data do each instance consist
of?}\label{what-data-does-each-instance-consist-of}}

The instances are individual, tabular data records in CSV format with
Demographic, Household and Family, Geographic, Financial, Work and
Education, Disability, and Survey Weights features.

In addition, there are metadata and documentation, \href{https://github.com/usnistgov/SDNist/blob/main/nist%20diverse%20communities%20data%20excerpts/national/national2019.json}{schema files} for each of the three geographies containing the features and valid data ranges in JSON format, and \href{https://github.com/usnistgov/SDNist/blob/main/nist%20diverse%20communities%20data%20excerpts/national/National%20Excerpts%20Postcard%20Descriptions.pdf}{`postcard' documentation} with English-language descriptions of the areas described by the data in PDF format. There are also an overarching \href{https://github.com/usnistgov/SDNist/tree/main/nist%20diverse%20communities%20data%20excerpts}{readme} and \href{https://github.com/usnistgov/SDNist/blob/main/nist%20diverse%20communities%20data%20excerpts/data_dictionary.json}{data dictionary}.

\hypertarget{is-there-a-label-or-target-associated-with-each-instance}{%
\subsubsection{Is there a label or target associated with each
instance?}\label{is-there-a-label-or-target-associated-with-each-instance}}

No.~These data are not designed specifically for classifier tasks.

\hypertarget{is-any-information-missing-from-individual-instances}{%
\subsubsection{Is any information missing from individual
instances?}\label{is-any-information-missing-from-individual-instances}}

There is no missing information in these excerpts, all records are
complete.

\hypertarget{are-relationships-between-individual-instances-made-explicit-e.g.-users-movie-ratings-social-network-links}{%
\subsubsection{Are relationships between individual instances made
explicit (e.g., users' movie ratings, social network
links)?}\label{are-relationships-between-individual-instances-made-explicit-e.g.-users-movie-ratings-social-network-links}}

Relationships between records have not been included in this version of
the data. Although the Excerpts data do contain multiple individuals
from the same household, it does not include the ACS PUMS Household ID
or relationship features needed to join them into a network. We expect
to include those features in a future update to the Excerpts.

\hypertarget{are-there-recommended-data-splits-e.g.-training-developmentvalidation-testing}{%
\subsubsection{Are there recommended data splits (e.g., training,
development/validation,
testing)?}\label{are-there-recommended-data-splits-e.g.-training-developmentvalidation-testing}}

There are three geographic partitions to facilitate benchmarking
algorithms on populations with differing levels of
heterogeneity/diversity (MA, TX and National). There are no splits
designed specifically for training and testing purposes. All of the data
presented at this time are from the 2019 ACS collection. In the future
we plan to add additional years.

\hypertarget{are-there-any-errors-sources-of-noise-or-redundancies-in-the-dataset}{%
\subsubsection{Are there any errors, sources of noise, or redundancies
in the
dataset?}\label{are-there-any-errors-sources-of-noise-or-redundancies-in-the-dataset}}

Although ACS data consumers generally assume the data remains
representative of the real population, the ACS PUMS data have had basic
statistical disclosure control deidentification applied (including
swapping and subsampling), which may impact its distribution. For more
information, see documentation from the
\href{https://www.census.gov/programs-surveys/acs/library/handbooks/general.html}{U.S.
Census Bureau.}

\hypertarget{is-the-dataset-self-contained-or-does-it-link-to-or-otherwise-rely-on-external-resources-e.g.-websites-tweets-other-datasets}{%
\subsubsection{Is the dataset self-contained, or does it link to or
otherwise rely on external resources (e.g., websites, tweets, other
datasets)?}\label{is-the-dataset-self-contained-or-does-it-link-to-or-otherwise-rely-on-external-resources-e.g.-websites-tweets-other-datasets}}

All of the data are self-contained within the repository. The data are
drawn from public domain sources and, thus, have no restrictions on usage.

\hypertarget{does-the-dataset-contain-data-that-might-be-considered-confidential}{%
\subsubsection{Does the dataset contain data that might be considered
confidential?}\label{does-the-dataset-contain-data-that-might-be-considered-confidential}}

The Excerpts are a subset
\href{https://www.census.gov/programs-surveys/acs/microdata/access.2019.html\#list-tab-735824205}{public
data published by the U.S. Census Bureau}. The U.S. Census Bureau is
bound by law, under Title 13 of the U.S. Code, to protect the identities
of individuals represented by the data.
\href{https://www.census.gov/about/policies/privacy/data_stewardship.html}{See
here} for details on the Census' data stewardship. The Census takes
elaborate steps to reduce risk of re-identification of individuals
surveyed and provide information regarding their suppression scheme
\href{https://www.census.gov/programs-surveys/acs/technical-documentation/data-suppression.html}{here}.

\hypertarget{does-the-dataset-contain-data-that-if-viewed-directly-might-be-offensive-insulting-threatening-or-might-otherwise-cause-anxiety}{%
\subsubsection{Does the dataset contain data that, if viewed directly,
might be offensive, insulting, threatening, or might otherwise cause
anxiety?}\label{does-the-dataset-contain-data-that-if-viewed-directly-might-be-offensive-insulting-threatening-or-might-otherwise-cause-anxiety}}

No.

\hypertarget{does-the-dataset-relate-to-people}{%
\subsubsection{Does the dataset relate to
people?}\label{does-the-dataset-relate-to-people}}

Yes.

\hypertarget{does-the-dataset-identify-any-subpopulations-e.g.-by-age-gender}{%
\subsubsection{Does the dataset identify any subpopulations (e.g., by
age,
gender)?}\label{does-the-dataset-identify-any-subpopulations-e.g.-by-age-gender}}

The data include demographic features such as Age, Sex, Race and
Hispanic Origin which may be used to disaggregate by subpopulation. It
additionally includes non-demographic features such as Educational
Attainment, Income Decile and Industry Category which also produce
subpopulation distributions with disparate patterns of feature
correlations.

The racial and ethnicity subpopulation breakdown by geography is as
follows (note that Hispanic origin and race are separate features): 
\begin{itemize}
    \item \texttt{MA Dataset (less diverse)}: 4\% Hispanic and 89\% White, 7\% Asian, 2\% Black, 2\% Other, 0\% AIANNH 
    \item \texttt{TX Dataset (more diverse)}: 19\% Hispanic
and 85\% White, 7\% Black, 4\% Other, 3\% Asian, 1\% AIANNH
    \item \texttt{National Dataset (especially diverse)}: 10\% Hispanic and 56\% White, 22\% Black, 10\% Other, 9\% Asian, 3\% AIANNH
\end{itemize}

\hypertarget{is-it-possible-to-identify-individuals-i.e.-one-or-more-natural-persons-either-directly-or-indirectly-i.e.-in-combination-with-other-data-from-the-dataset}{%
\subsubsection{Is it possible to identify individuals (i.e., one or more
natural persons), either directly or indirectly (i.e., in combination
with other data) from the
dataset?}\label{is-it-possible-to-identify-individuals-i.e.-one-or-more-natural-persons-either-directly-or-indirectly-i.e.-in-combination-with-other-data-from-the-dataset}}

The Excerpts are survey results from real individuals as collected by
the U.S. Census Bureau.
\href{\%5B\#\#\#\%20Does\%20the\%20dataset\%20contain\%20data\%20that\%20might\%20be\%20considered\%20confidential?}{See
response above} for more information about Census' data protections.

The subset of the Census' data that we provide here introduces no
additional information, and, therefore, does not increase the risk of
identifying individuals.

\hypertarget{does-the-dataset-contain-data-that-might-be-considered-sensitive-in-any-way-e.g.-data-that-reveals-racial-or-ethnic-origins-sexual-orientations-religious-beliefs-political-opinions-or-union-memberships-or-locations-financial-or-health-data-biometric-or-genetic-data-forms-of-government-identification-such-as-social-security-numbers-criminal-history}{%
\subsubsection{Does the dataset contain data that might be considered
sensitive in any way (e.g., data that reveals racial or ethnic origins,
sexual orientations, religious beliefs, political opinions or union
memberships, or locations; financial or health data; biometric or
genetic data; forms of government identification, such as social
security numbers; criminal
history)?}\label{does-the-dataset-contain-data-that-might-be-considered-sensitive-in-any-way-e.g.-data-that-reveals-racial-or-ethnic-origins-sexual-orientations-religious-beliefs-political-opinions-or-union-memberships-or-locations-financial-or-health-data-biometric-or-genetic-data-forms-of-government-identification-such-as-social-security-numbers-criminal-history}}

Yes. These data are detailed demographic records.
\href{\%5B\#\#\#\%20Does\%20the\%20dataset\%20contain\%20data\%20that\%20might\%20be\%20considered\%20confidential?}{See
response above} for more information about Census Bureau's data
protections.

\hypertarget{any-other-comments-1}{%
\subsubsection{Any other comments?}\label{any-other-comments-1}}

No.

\hypertarget{collection-process}{%
\subsection{Collection process}\label{collection-process}}

\hypertarget{how-was-the-data-associated-with-each-instance-acquired}{%
\subsubsection{How was the data associated with each instance
acquired?}\label{how-was-the-data-associated-with-each-instance-acquired}}

These data is a curated geographic subsample of the 2019 American
Community Survey Public Use Microdata files. The U.S. Census Bureau
details its survey data collection approach
\href{https://www.census.gov/programs-surveys/acs/library/handbooks/general.html}{here}.

\hypertarget{what-mechanisms-or-procedures-were-used-to-collect-the-data-e.g.-hardware-apparatus-or-sensor-manual-human-curation-software-program-software-api}{%
\subsubsection{What mechanisms or procedures were used to collect the
data (e.g., hardware apparatus or sensor, manual human curation,
software program, software
API)?}\label{what-mechanisms-or-procedures-were-used-to-collect-the-data-e.g.-hardware-apparatus-or-sensor-manual-human-curation-software-program-software-api}}

See previous response.

\hypertarget{if-the-dataset-is-a-sample-from-a-larger-set-what-was-the-sampling-strategy-e.g.-deterministic-probabilistic-with-specific-sampling-probabilities}{%
\subsubsection{If the dataset is a sample from a larger set, what was
the sampling strategy (e.g., deterministic, probabilistic with specific
sampling
probabilities)?}\label{if-the-dataset-is-a-sample-from-a-larger-set-what-was-the-sampling-strategy-e.g.-deterministic-probabilistic-with-specific-sampling-probabilities}}

The data set is a (deterministic) curated sample by geography. It is
comprised of 31 Public Use Microdata Areas 

\begin{itemize}
\item \texttt{national}: 27254
records drawn from 20 Public Use Microdata Areas (PUMAs) from across the
United States. This excerpt was selected to include communities with
very diverse subpopulation distributions. 
\item \texttt{texas}: 9276 records
drawn from six PUMAs of communities surrounding Dallas-Fort Worth, Texas
area. This excerpt was selected to focus on areas with moderate
diversity. 
\item \texttt{massachusetts}: 7634 records drawn from five PUMAs
of communities from the North Shore to the west of the greater Boston,
Massachusetts area. This excerpt was selected to focus on areas with
less diversity.
\end{itemize}

\hypertarget{who-was-involved-in-the-data-collection-process-e.g.-students-crowdworkers-contractors-and-how-were-they-compensated-e.g.-how-much-were-crowdworkers-paid}{%
\subsubsection{Who was involved in the data collection process (e.g.,
students, crowdworkers, contractors) and how were they compensated
(e.g., how much were crowdworkers
paid)?}\label{who-was-involved-in-the-data-collection-process-e.g.-students-crowdworkers-contractors-and-how-were-they-compensated-e.g.-how-much-were-crowdworkers-paid}}

{[}See response above.{]}

\hypertarget{over-what-timeframe-was-the-data-collected}{%
\subsubsection{Over what timeframe was the data
collected?}\label{over-what-timeframe-was-the-data-collected}}

These data was collected during 2019.

\hypertarget{were-any-ethical-review-processes-conducted-e.g.-by-an-institutional-review-board}{%
\subsubsection{Were any ethical review processes conducted (e.g., by an
institutional review
board)?}\label{were-any-ethical-review-processes-conducted-e.g.-by-an-institutional-review-board}}

The Excerpts are a curated subsample of existing public data published
by the U.S. Government. No internal review board (IRB)review was necessary by institution
policy.

\hypertarget{does-the-dataset-relate-to-people-1}{%
\subsubsection{Does the dataset relate to
people?}\label{does-the-dataset-relate-to-people-1}}

Yes.

\hypertarget{did-you-collect-the-data-from-the-individuals-in-question-directly-or-obtain-it-via-third-parties-or-other-sources-e.g.-websites}{%
\subsubsection{Did you collect the data from the individuals in question
directly, or obtain it via third parties or other sources (e.g.,
websites)?}\label{did-you-collect-the-data-from-the-individuals-in-question-directly-or-obtain-it-via-third-parties-or-other-sources-e.g.-websites}}

Other Sources. These data are a curated geographic subsample of the 2019
American Community Survey Public Use Microdata files, which are
\href{https://www.census.gov/programs-surveys/acs/microdata/access.2019.html\#list-tab-735824205}{available
here}.

\hypertarget{were-the-individuals-in-question-notified-about-the-data-collection}{%
\subsubsection{Were the individuals in question notified about the data
collection?}\label{were-the-individuals-in-question-notified-about-the-data-collection}}

Yes.
\protect\hyperlink{ux5cux23ux5cux23ux5cux2520Howux5cux2520wasux5cux2520theux5cux2520dataux5cux2520associatedux5cux2520withux5cux2520eachux5cux2520instanceux5cux2520acquiredux3f}{See
response above.}

\hypertarget{did-the-individuals-in-question-consent-to-the-collection-and-use-of-their-data}{%
\subsubsection{Did the individuals in question consent to the collection
and use of their
data?}\label{did-the-individuals-in-question-consent-to-the-collection-and-use-of-their-data}}

Yes.
\protect\hyperlink{ux5cux23ux5cux23ux5cux2520Howux5cux2520wasux5cux2520theux5cux2520dataux5cux2520associatedux5cux2520withux5cux2520eachux5cux2520instanceux5cux2520acquiredux3f}{See
response above.}

\hypertarget{if-consent-was-obtained-were-the-consenting-individuals-provided-with-a-mechanism-to-revoke-their-consent-in-the-future-or-for-certain-uses}{%
\subsubsection{If consent was obtained, were the consenting individuals
provided with a mechanism to revoke their consent in the future or for
certain
uses?}\label{if-consent-was-obtained-were-the-consenting-individuals-provided-with-a-mechanism-to-revoke-their-consent-in-the-future-or-for-certain-uses}}

\protect\hyperlink{ux5cux23ux5cux23ux5cux2520Howux5cux2520wasux5cux2520theux5cux2520dataux5cux2520associatedux5cux2520withux5cux2520eachux5cux2520instanceux5cux2520acquiredux3f}{See
response above.}

\hypertarget{has-an-analysis-of-the-potential-impact-of-the-dataset-and-its-use-on-data-subjects-e.g.-a-data-protection-impact-analysis-been-conducted}{%
\subsubsection{Has an analysis of the potential impact of the dataset
and its use on data subjects (e.g., a data protection impact analysis)
been
conducted?}\label{has-an-analysis-of-the-potential-impact-of-the-dataset-and-its-use-on-data-subjects-e.g.-a-data-protection-impact-analysis-been-conducted}}

These data is a curated geographic subsample of the 2019 American
Community Survey (ACS) Public Use Microdata files. Many investigations
have examined ACS data with some information
\href{https://www.census.gov/programs-surveys/acs/library/handbooks/general.html}{published
by the Census Bureau itself}.

The data presented here, the Excerpts, are a subset of the data and
present no additional risks to the subjects surveyed by the Census.

\hypertarget{any-other-comments-2}{%
\subsubsection{Any other comments?}\label{any-other-comments-2}}

No.

\hypertarget{preprocessingcleaninglabeling}{%
\subsection{Preprocessing/cleaning/labeling}\label{preprocessingcleaninglabeling}}

\hypertarget{was-any-preprocessingcleaninglabeling-of-the-data-done-e.g.-discretization-or-bucketing-tokenization-part-of-speech-tagging-sift-feature-extraction-removal-of-instances-processing-of-missing-values}{%
\subsubsection{Was any preprocessing/cleaning/labeling of the data done
(e.g., discretization or bucketing, tokenization, part-of-speech
tagging, SIFT feature extraction, removal of instances, processing of
missing
values)?}\label{was-any-preprocessingcleaninglabeling-of-the-data-done-e.g.-discretization-or-bucketing-tokenization-part-of-speech-tagging-sift-feature-extraction-removal-of-instances-processing-of-missing-values}}

The original ACS data are clean, and no class labeling was done. However,
several Excerpts features are new derivatives of ACS features designed
to provide easier access to certain information. Population DENSITY
divides PUMA population by surface area and allows models to distinguish
rural and urban geographies. INDP\_CAT aggregates detailed industry
codes into a small set of broad categories. PINCP\_DECILE aggregates
incomes into percentile bins relative to the record's state. And, EDU
simplifies the original ACS schooling feature to focus on milestone
grades and degrees.

\hypertarget{was-the-raw-data-saved-in-addition-to-the-preprocessedcleanedlabeled-data-e.g.-to-support-unanticipated-future-uses}{%
\subsubsection{Was the ``raw'' data saved in addition to the
preprocessed/cleaned/labeled data (e.g., to support unanticipated future
uses)?}\label{was-the-raw-data-saved-in-addition-to-the-preprocessedcleanedlabeled-data-e.g.-to-support-unanticipated-future-uses}}

See the U.S. Census Bureau's documentation for information about
\href{https://www.census.gov/programs-surveys/acs/library/handbooks/general.html}{published
ACS data}.

\hypertarget{is-the-software-used-to-preprocesscleanlabel-the-instances-available}{%
\subsubsection{Is the software used to preprocess/clean/label the
instances
available?}\label{is-the-software-used-to-preprocesscleanlabel-the-instances-available}}

The preprocessing was minimal (addition of a small set of derivative
features), and can be reproduced as described above. The code is not
currently available.

\hypertarget{any-other-comments-3}{%
\subsubsection{Any other comments?}\label{any-other-comments-3}}

No.

\hypertarget{uses}{%
\subsection{Uses}\label{uses}}

\hypertarget{has-the-dataset-been-used-for-any-tasks-already}{%
\subsubsection{Has the dataset been used for any tasks
already?}\label{has-the-dataset-been-used-for-any-tasks-already}}

The Excerpts serve as benchmark data for two open source projects at the
National Institute of Standards and Technology (NIST): the
\href{https://doi.org/10.18434/mds2-2943}{SDNist Deidentified Data
Report tool} and the
\href{https://pages.nist.gov/privacy_collaborative_research_cycle/}{2023
Collaborative Research Cycle} (CRC).

\hypertarget{is-there-a-repository-that-links-to-any-or-all-papers-or-systems-that-use-the-dataset}{%
\subsubsection{Is there a repository that links to any or all papers or
systems that use the
dataset?}\label{is-there-a-repository-that-links-to-any-or-all-papers-or-systems-that-use-the-dataset}}

No.~Users are not mandated to contribute their work to any central
repository. We publish user-contributed data
\href{https://github.com/usnistgov/privacy_collaborative_research_cycle}{here}.
We recommend that data users cite our work using the
\href{https://doi.org/10.18434/mds2-2895}{dataset DOI}.

\hypertarget{what-other-tasks-could-the-dataset-be-used-for}{%
\subsubsection{What (other) tasks could the dataset be used
for?}\label{what-other-tasks-could-the-dataset-be-used-for}}

The Excerpts were designed for benchmarking privacy-preserving data
deidentification techniques such as synthetic data or statistical
disclosure limitation. However, they can be used to study the behavior
of any tabular data machine learning or analysis technique when applied
to diverse populations. Synthetic data generators are just an especially
verbose application of machine learning (producing full records rather
than class labels), so tools designed to improve understanding of
synthetic data have potential for a much broader application.

\hypertarget{is-there-anything-about-the-composition-of-the-dataset-or-the-way-it-was-collected-and-preprocessedcleanedlabeled-that-might-impact-future-uses}{%
\subsubsection{Is there anything about the composition of the dataset or
the way it was collected and preprocessed/cleaned/labeled that might
impact future
uses?}\label{is-there-anything-about-the-composition-of-the-dataset-or-the-way-it-was-collected-and-preprocessedcleanedlabeled-that-might-impact-future-uses}}

The U.S. Census Bureau recommends using sampling weights to account for
survey undersampling and generate equitable full population statistics.
The PWGPT feature included in the Excerpts is the person (record) level
sampling weight. For full population statistics, each record should be
multiplied by its sampling weight.

\hypertarget{are-there-tasks-for-which-the-dataset-should-not-be-used}{%
\subsubsection{Are there tasks for which the dataset should not be
used?}\label{are-there-tasks-for-which-the-dataset-should-not-be-used}}

The Excerpts are suitable for any application relevant to government
survey data over the selected feature set.

\hypertarget{any-other-comments-4}{%
\subsubsection{Any other comments?}\label{any-other-comments-4}}

No.

\hypertarget{distribution}{%
\subsection{Distribution}\label{distribution}}

\hypertarget{will-the-dataset-be-distributed-to-third-parties-outside-of-the-entity-e.g.-company-institution-organization-on-behalf-of-which-the-dataset-was-created}{%
\subsubsection{Will the dataset be distributed to third parties outside
of the entity (e.g., company, institution, organization) on behalf of
which the dataset was
created?}\label{will-the-dataset-be-distributed-to-third-parties-outside-of-the-entity-e.g.-company-institution-organization-on-behalf-of-which-the-dataset-was-created}}

\hypertarget{has-the-dataset-been-used-for-any-tasks-already-1}{%
\subsubsection{Has the dataset been used for any tasks
already?}\label{has-the-dataset-been-used-for-any-tasks-already-1}}

The Excerpts serve as benchmark data for two open source projects at the
National Institute of Standards and Technology (NIST): the
\href{https://doi.org/10.18434/mds2-2943}{SDNist Deidentified Data
Report tool} and the
\href{https://pages.nist.gov/privacy_collaborative_research_cycle/}{2023
Collaborative Research Cycle} (CRC).

\hypertarget{how-will-the-dataset-will-be-distributed-e.g.-tarball-on-website-api-github}{%
\subsubsection{How will the dataset will be distributed (e.g., tarball
on website, API,
GitHub)?}\label{how-will-the-dataset-will-be-distributed-e.g.-tarball-on-website-api-github}}

\href{https://doi.org/10.18434/mds2-2895}{10.18434/mds2-289}

\hypertarget{when-will-the-dataset-be-distributed}{%
\subsubsection{When will the dataset be
distributed?}\label{when-will-the-dataset-be-distributed}}

The dataset is currently available to the public.

\hypertarget{will-the-dataset-be-distributed-under-a-copyright-or-other-intellectual-property-ip-license-andor-under-applicable-terms-of-use-tou}{%
\subsubsection{Will the dataset be distributed under a copyright or
other intellectual property (IP) license, and/or under applicable terms
of use
(ToU)?}\label{will-the-dataset-be-distributed-under-a-copyright-or-other-intellectual-property-ip-license-andor-under-applicable-terms-of-use-tou}}

The data are in the public domain.
\href{https://www.nist.gov/open/copyright-fair-use-and-licensing-statements-srd-data-software-and-technical-series-publications}{See
the following statement from NIST}.

\hypertarget{have-any-third-parties-imposed-ip-based-or-other-restrictions-on-the-data-associated-with-the-instances}{%
\subsubsection{Have any third parties imposed IP-based or other
restrictions on the data associated with the
instances?}\label{have-any-third-parties-imposed-ip-based-or-other-restrictions-on-the-data-associated-with-the-instances}}

No.~All data are drawn from public domain sources.

\hypertarget{do-any-export-controls-or-other-regulatory-restrictions-apply-to-the-dataset-or-to-individual-instances}{%
\subsubsection{Do any export controls or other regulatory restrictions
apply to the dataset or to individual
instances?}\label{do-any-export-controls-or-other-regulatory-restrictions-apply-to-the-dataset-or-to-individual-instances}}

No.~All data are drawn from public domain sources and have no known
export or regulatory restrictions.

\hypertarget{any-other-comments-5}{%
\subsubsection{Any other comments?}\label{any-other-comments-5}}

No.

\hypertarget{maintenance}{%
\subsection{Maintenance}\label{maintenance}}

\hypertarget{who-is-supportinghostingmaintaining-the-dataset}{%
\subsubsection{Who is supporting/hosting/maintaining the
dataset?}\label{who-is-supportinghostingmaintaining-the-dataset}}

This dataset is hosted by NIST and maintained by the Privacy Engineering
Program.

\hypertarget{how-can-the-ownercuratormanager-of-the-dataset-be-contacted-e.g.-email-address}{%
\subsubsection{How can the owner/curator/manager of the dataset be
contacted (e.g., email
address)?}\label{how-can-the-ownercuratormanager-of-the-dataset-be-contacted-e.g.-email-address}}

Dataset managers can be reached by
\href{https://github.com/usnistgov/SDNist/issues}{raising an issue},
emailing the \href{PrivacyEng@nist.gov}{Privacy Engineering Program}, or
by contacting the project principal investigator,
\href{mailto:gary.howarth@nist.gov}{Gary Howarth}.

\hypertarget{is-there-an-erratum}{%
\subsubsection{Is there an erratum?}\label{is-there-an-erratum}}

There have been small updates to the meta-data data dictionary.json
files (for example, to improve clarity in descriptive strings for
features). The data are maintained in a public GIt repository and, thus,
all changes to the data are recorded in a public ledger.

\hypertarget{will-the-dataset-be-updated-e.g.-to-correct-labeling-errors-add-new-instances-delete-instances}{%
\subsubsection{Will the dataset be updated (e.g., to correct labeling
errors, add new instances, delete
instances)?}\label{will-the-dataset-be-updated-e.g.-to-correct-labeling-errors-add-new-instances-delete-instances}}

Since the data are excerpts from the 2019 release of the American
Community Survey, we do not expect any updates to labels or instances.
We do plan on one mayor updated version release in the future with the
following improvements: 
\begin{itemize}
\item \texttt{Household ID features}: Allows joins
between individuals in the same household 
\item \texttt{Individual ID}: Supports
reidentification research. 
\item \texttt{Training Data Partition}: Including
excerpts from 2018 for algorithm development/training and as a baseline
for reidentification studies 
\item \texttt{Large-sized low-diversity excerpt}: Our
current low-diversity excerpts, MA and TX, have much fewer records than
our high-diversity excerpt, National; this can be a confounding factor
for comparative analyses.
\end{itemize}

\hypertarget{if-the-dataset-relates-to-people-are-there-applicable-limits-on-the-retention-of-the-data-associated-with-the-instances-e.g.-were-individuals-in-question-told-that-their-data-would-be-retained-for-a-fixed-period-of-time-and-then-deleted}{%
\subsubsection{If the dataset relates to people, are there applicable
limits on the retention of the data associated with the instances (e.g.,
were individuals in question told that their data would be retained for
a fixed period of time and then
deleted)?}\label{if-the-dataset-relates-to-people-are-there-applicable-limits-on-the-retention-of-the-data-associated-with-the-instances-e.g.-were-individuals-in-question-told-that-their-data-would-be-retained-for-a-fixed-period-of-time-and-then-deleted}}

These data are in the public domain and as such there are no retention
limits.

\hypertarget{will-older-versions-of-the-dataset-continue-to-be-supportedhostedmaintained}{%
\subsubsection{Will older versions of the dataset continue to be
supported/hosted/maintained?}\label{will-older-versions-of-the-dataset-continue-to-be-supportedhostedmaintained}}

The data are maintained in a public Git repository and, thus, all changes
to the data are recorded in a public ledger. There are specific releases
in the repository that capture major data milestones.

\hypertarget{if-others-want-to-extendaugmentbuild-oncontribute-to-the-dataset-is-there-a-mechanism-for-them-to-do-so}{%
\subsubsection{If others want to extend/augment/build on/contribute to
the dataset, is there a mechanism for them to do
so?}\label{if-others-want-to-extendaugmentbuild-oncontribute-to-the-dataset-is-there-a-mechanism-for-them-to-do-so}}

We invite the public to use and build on these resources. First, these
resources are provided by NIST as a public service, and the public is
free to integrate these resources into their own work. Second, we invite
the public to raise issues in the dataset repository, allowing for a
transparent interaction. Individuals and groups wishing to make
substantial contributions are encouraged to contact the project
principal investigator, \href{mailto:gary.howarth@nist.gov}{Gary
Howarth}.

\hypertarget{any-other-comments-6}{%
\subsubsection{Any other comments?}\label{any-other-comments-6}}

No.

\section{Math Appendix}
\subsection{Proofs of Lemmas 2.1 and 2.2}
We introduced the concept of dispersal ratio in the main paper with the purpose of a giving the reader a clear and intuitive explanation of the term. In doing so, we omitted some formal results that might be interesting to examine in order to understand the mechanics behind dispersal ratio and independence. The perceptive reader may have noticed that we stated two lemmas in Section 2 without proving them. Recall the definition of dispersal ratio.
\begin{definition}
[Dispersal Ratio] Let the dispersal ratio for a population $P$ with the addition of feature $X$ be defined as 
$$Disperse(S, X, P) =  |bin_{(S+X)(P)}|  / |bin_{S(P)} |$$
\end{definition}
We begin by providing proofs of Lemma 2.1 (corresponding to Lemma C.1) and Lemma 2.2 (corresponding to Lemma C.2) as stated in Section 2. We follow it up with a result that may be of interest in Section C.3. These proofs follows the same framework and terminology as used in the main paper.

\begin{lemma}
An uncertainty coefficient of 1 is equivalent to a dispersal ratio of 1.
\[U(X|F) = 1 \iff Disperse(S, X, P) = 1\]  
\end{lemma}
\begin{proof}
 \[U(X|F) = 1 \\
 \Rightarrow {\frac{H(X)-H(X|F)}{H(X)}} = 1 \\
 \Rightarrow H(X|F) = 0
 \] \\
Consider the following result. $H(X|F)=0$ if and only if $X$ is a function of $F$ i.e., $\forall f: p(f)>0$, there is only one possible value of $x$ with $p(x,f)>0$ \cite{cover1999elements}.

Let the function $g$ between $X$ and $F$ be denoted by $F = g(X)$. Applying the result here, let there be $m$ elements in the domain of $F$, which implies there can be no more than $m$ elements in the co-domain of $X$, to constitute a valid function. Let the elements in the range of $F$ be denoted by $f_1,f_2... f_m$, and that of $X$ be denoted by $x_1,x_2... x_{m'}$.\\
Since $X$ is a function of $F$, there exists only one element $x_i \in X$ corresponding to $f_j \in F$. \\
Thus, all bins in the schema $(S+x)$ can be denoted by $(f_i,x_i) = (f_i,g(f_i))$. Since there are $m$ bins, corresponding the size of the domain, in $F$, there will be exactly $m$ bins in $F' = (F,X)$. Therefore, 
\[|bin_{(S+X)(P)}| = |bin_{S(P)}|\]
\[\Rightarrow Disperse(S,X,P) = 1\]
Similarly, the converse of the lemma can be proved by taking the converse of the above result and considering that the inverse of the function $g' = g^{-1}(x)$ for $x$: $(F,X) \rightarrow F$ is uniquely defined if the dispersal ratio is 1.
\end{proof}

\begin{lemma}
An uncertainty coefficient of 0 leads to the maximum dispersal ratio.
\[U(X|F) = 0 \implies Disperse(S, X, P) = |Range(X)|\]
\end{lemma}
\begin{proof}

\[U(X|F) = 0 
\implies {\frac{H(X)-H(X|F)}{H(X)}} = 0
\implies H(X|F) = H(X)
\] 
This implies $X$ and $F$ are independent observations \cite{cover1999elements}. Observe that the range of $Y = (X,F)$ can take maximum $ n_{max} = |Range(X)||Range(F)|$ values since it is the number of elements in $X \times F$. Note that $|Range(F)| = |bin_{S(P)}|$. Here, $Range(Y) = n_{max}$ due to the independence of $X$ and $F$ since
\[\forall x,f: Pr[Y = y] = (Pr[X=x] * Pr[F=f]) \neq 0\] 
As there are $n_{max}$ non-zero values for the probability distribution of $Y$, the size of the range of $Y$ is maximum. Note that $Y$ exactly expresses the distribution of values in the schema $(S+X)$. \\
 Therefore,
\[|bin_{(S+X)(P)}| = |bin_{S(P)}||Range(X)|\]
\[\implies Disperse(S,X,P) = |Range(X)|\]
which is the maximum dispersal ratio since $|bin_{S(P)}|*|Range(X)|$ was maximized.
\end{proof}

\subsection{Proofs of Theorems 2.3 and 2.4}
Here, we provide the detailed proofs of Theorem 2.3 (corresponding to Lemma C.3) and Theorem 2.4 (corresponding to Lemma C.4) as stated in Section 2. These proofs follows the same framework and terminology as used in the main paper.
\begin{theorem} Dispersal Ratio is bounded from above and below as function of the independence of the added feature as follows
\[\frac{|P|\cdot f(u)}{\log(|P|)|Range(F)|}\ge Disperse(S,X,P)\ge \frac{2^{f(u)}}{|Range(F)|}\]
where $f(u) := (1-u)H(X) + H(F)$ with $u = U(X|F)$.
\end{theorem}
\begin{proof}
Lemma 2.1 and Lemma 2.2 give an upper and lower bound for the dispersal ratio which is
\[1 \leq Disperse(S,X,P) \leq |Range(X)|\]
corresponding to
\[1 \geq U(X|F) \geq 0 \]

Independence is quantified through the uncertainty coefficient $U(X|F)$. Recall that as $U$ decreases, independence increases, and vice-versa. 

Let some arbitrary $u = U(X|F)$. From the definition of $U(X|F)$,
\[u = {\frac{H(X)-H(X|F)}{H(X)}} \Rightarrow H(X|F) = (1-u)H(X)\]
Rewriting in terms of the joint entropy \cite{cover1999elements} $H(X,F)$,
\begin{equation}
H(X,F) = (1-u)H(X) + H(F)
\end{equation}
Applying a well-known upper bound on entropy \cite{cover1999elements}, and substituting in equation (1),
\begin{equation}
H(X,F) \le \log_2(|Range(X,F)|) \implies |Range(X,F)| \ge 2^{(1-u)H(X) + H(F)}
\end{equation}

Following from our definition of dispersal ratio,
\begin{equation}
Disperse(S,X,P) = \frac{|bin_{(S+X)(P)}|}{|bin_{(S)(P)}|} = \frac{|Range(X,F)|}{|Range(F)|}
\end{equation}
Thus,
\begin{equation}
Disperse(S,X,P)\ge \frac{2^{(1-u)H(X) + H(F)}}{|Range(F)|}
\end{equation}

Now, from the definition of entropy,
\[ \displaystyle \mathrm {H} (X,F) =-\sum _{x\in {\mathcal {(X \times  F)}}}p(x)\log p(x)\]
Since each bin, corresponding to $x$ in the above equation, must have at least one person in order to contribute to the entropy, $p(x) \ge \frac{1}{|P|}$ where $|P|$ is the size of the population. To analyze any arbitrary distribution, we can first allocate one person to each bin in the range, by definition of range. Now, we have to distribute $|P| - |Range(X,F)|$ people in $|Range(X,F)|$ bins. Entropy is minimized when all the rest of the people are put in one bin. This distribution gives a lower bound for entropy in terms of the range of $(X,F)$. Formally, we get the following inequality,

\begin{equation}
% \begin{split}
\scalebox{0.95}{
% \displaystyle \mathrm
$ %need inline mathmode to use scalebox
{H} (X,F) \ge (|Range(X,F)|-1) \left[ \frac{\log(|P|)}{|P|} \right] \\ + \frac{|P|-(|Range(X,F)|-1)}{|P|}\log\left(\frac{|P|}{|P|-(|Range(X,F)|-1)}\right)
$ %need inline mathmode to use scalebox
}%end scale box
% \end{split}
\end{equation}

Since $|P| >> |Range(X,F)| >> 1$, the second term is neglected as $1 \cdot \log(\frac{|P|}{|P|}) = 0$,
\[\implies \displaystyle \mathrm {H} (X,F) \ge |Range(X,F)| \left[ \frac{\log(|P|)}{|P|} \right]\]

Substituting in equation (1) and rearranging terms, we get
\begin{equation}
Disperse(S,X,P) \le \frac{|P| \cdot ((1-u)H(X) + H(F))}{\log(|P|)|Range(F)|}
\end{equation}

We can improve on our trivial bounds of 1 and $|Range(X)|$, from Lemma 2.1 and Lemma 2.2, for the dispersal ratio corresponding to a given $u$. Combining these two results with the improved upper and lower bounds from equation (3) and equation (5),
\begin{equation}
% \begin{split}
\scalebox{0.95}{
$
\min\left\{\frac{|P| \cdot ((1-u)H(X) + H(F))}{\log(|P|)|Range(F)|},|Range(X)|\right\}\ge 
\\
Disperse(S,X,P) \ge \max\left\{\frac{2^{(1-u)H(X) + H(F)}}{|Range(F)|},1\right\} 
$
}
% \end{split}
\end{equation}

From equation (1), we can write $H(X,F)$ as a function of $u$ i.e., $f(u) := (1-u)H(X) + H(F)$. Also note that as $u$ increases, $f(u)$ decreases and vice-versa. This gives a simpler form for our above equation (also considering only non-trivial bounds),
\begin{equation}
\frac{|P|\cdot f(u)}{\log(|P|)|Range(F)|}\ge Disperse(S,X,P)\ge \frac{2^{f(u)}}{|Range(F)|}
\end{equation}
\end{proof}
This result shows that, for some fixed value of entropy of the added feature, the non-trivial upper and lower bounds for the dispersal ratio decrease as the uncertainty coefficient increases, and vice-versa, according to the described behaviour of $f(u)$.

Now, we want to compare the effect of adding a new feature $X_1$ or $X_2$ to the schema. Let us assume that $X_1$ is more "independent" than $X_2$ of the distribution of $P$ in $S$ (note that this is equivalent to considering a single feature $X$ and two diverse subpopulations $P1, P2$ with differing relationships to $X$). We can then say that the uncertainty coefficient $u_1 = U(X_1|F)$ is lesser for $X_1$ as compared to $u_2$ corresponding to $X_2$. We can use our results from the above theorem to make this comparison as follows in Theorem 2.4. We define a couple of terms first for ease of notation.

Let the non-trivial lower bound for the dispersal ratio (>1) on adding feature $X$ be denoted as
\begin{equation}
LB(Disperse(S,X,P)) = \frac{2^{(1-u)H(X) + H(F)}}{|Range(F)|}
\end{equation}

Let the non-trivial upper bound for the dispersal ratio (<$|Range(X)|$) on adding feature $X$ be denoted
\begin{equation}
UB(Disperse(S,X,P)) = \frac{|P|\cdot (1-u)H(X) + H(F)}{\log(|P|)|Range(F)|}
\end{equation}

\begin{theorem} Consider two features $X_1$ and $X_2$, identical in terms of entropy, that can be added to the schema. If $X_1$ has higher independence than $X_2$ with respect to $F$, it is equivalent to $X_1$ having a higher LB and higher UB for the dispersal ratio.
\[U(X_1|F) \le U(X_2|F) \iff LB(Disperse(S,X_1,P) \ge LB(Disperse(S,X_2,P) \]
\[U(X_1|F) \le U(X_2|F) \iff UB(Disperse(S,X_1,P) \ge UB(Disperse(S,X_2,P)\]
\end{theorem}
\begin{proof}
Consider a population $P$ distributed in a table-based partitioned schema $S$. Note that higher independence corresponds to a lower uncertainty coefficient.
Let $u_1 = U(X_1|F)$, $u_2 = U(X_2|F)$ and $H(X) = H(X_1) = H(X_2)$. The following result can be derived from Theorem 2.3. Let us consider the lower bound first. From equation (9),

\[\frac{LB((Disperse(S,X_1,P))}{LB((Disperse(S,X_2,P))} = \frac{\frac{2^{(1-u_1)H(X) + H(F)}}{|Range(F)|}}{\frac{2^{(1-u_2)H(X) + H(F)}}{|Range(F)|}}
\implies \frac{LB((Disperse(S,X_1,P))}{LB((Disperse(S,X_2,P))} = 2^{H(X)(u_2-u_1)}\]

Since entropy is always greater than or equal to 0 , $H(X) \ge 0$ \cite{cover1999elements}. Thus,
\[LB(Disperse(S,X_1,P) \ge LB(Disperse(S,X_2,P) \iff u_2 - u_1 \ge 0
\iff u_1 \le u_2\]

Similarly for upper bound, from equation (10), we get
\[\frac{UB((Disperse(S,X_1,P))}{UB((Disperse(S,X_2,P))} = \frac{(1-u_1)H(X) + H(F)}{(1-u_2)H(X) + H(F)}\]
Clearly,
\[UB(Disperse(S,X_1,P) \ge UB(Disperse(S,X_2,P) \iff (1-u_1)H(X) + H(F) \ge (1-u_2)H(X) + H(F)\]
\[\iff u_1 \le u_2\]
\end{proof}

\subsection{Additional Material}
We now show an interesting consequence of the relation between dispersal ratio and the initial population. The following lemmas prove that a small population size can lead to small cell counts.

Consider a population $P$ with a sub-population $P_1$, distributed in a table-based partitioned schema. Consider an individual $i \in P_1$, who gets placed in a bin under schema $S$. We denote the size of that bin as $size(bin_{S(i)})$. Let a feature $f$ be added to the schema.

\begin{lemma}
The dispersal ratio is always greater than or equal to 1. 
\[Disperse(S, f, P_1) \ge 1\]
\end{lemma}
\begin{proof}
Consider an arbitrary $bin_S(i)$ in the schema $S$ with the $m$ features in the feature set $f_1,f_2,f_3...f_m$. 

Adding a new feature $f$ to the schema $S$ with feature values (say) in the set $V= \{v_1,v_2\}$ will subdivide all records in $f_1,f_2,f_3...f_m$ into  $f_1,f_2,f_3...f_m,v_1$ and $f_1,f_2,f_3...f_m,v_2$, by the definition of partitioning.

$bin_S(i)$ in the schema $S$ will be replaced by at least one bin or more, in the schema $(S+f)$. Thus, the dispersal ratio for the sub-population of $bin_{S(i)}: i \in {P_1}$ is always greater than 1.

Since for each disjoint sub-population corresponding to each bin $\in S$, this ratio is greater than one, the dispersal ratio for the overall population $P_1$ over the schema $S$ and adding a new feature $f$, is also greater than 1.
\end{proof}

\begin{definition}
[Average bin size for population $P_1$] It is defined as 
\[\left[\sum_{S(i): i \in P_1}size(bin_{S(i)})\right]/{|bin_{S(P_1)}|}\]
\end{definition}

\begin{lemma}
If a new feature f is added to the schema denoted by $S+f$, then the average bin size will stay the same or decrease.
\[\left[\sum_{S(i): i \in P_1}size(bin_{S(i)})\right]/{|bin_{S(P_1)}|} \ge \left[\sum_{(S+f)(i): i \in P_1}size(bin_{(S+f)(i)})\right]/{|bin_{(S+f)(P_1)}|}\]
\end{lemma}
\begin{proof}
For each of the disjoint partitions of some $S(i): i \in P_1$, records of the form $i \in P_1$ do not get merged with any records that were not in the initial bin $S(i)$, by definition of partitioning. Thus, summing over all such bins,
\begin{equation}
\left[\sum_{S(i): i \in P_1}size(bin_{S(i)})\right] \ge \left[\sum_{(S+f)(i): i \in P_1}size(bin_{(S+f)(i)})\right]    
\end{equation}
Note that there is a '$\ge$' inequality since there may be bins in the schema $S+f$ that do contain records of the form $i \in P_1$, which were previously grouped with records $i \in P_1$ in the schema $S$. From Lemma C.3, if the dispersal ratio for population $P_1$ is $r_1$, then $r_1\geq1$, which implies

\begin{equation}
|bin_{S(P_1)}| \leq |bin_{(S+f)(P_1)}|    
\end{equation}
Combining equations (1) and (2) proves our result, by observing that they are the numerator and denominator respectively of our desired inequality.
\end{proof}

Assume two sub-populations $P_0$ and $P_1$ are distributed in the same arbitrary number of bins $|bin_{S(P_1)}| = |bin_{S(P_0)}| = m$. If on adding a feature $f$, $P_0$ and $P_1$ have the same dispersal ratio ($r_0 = r_1 = r'$), then $|bin_{(S+f)(P_1)}| = |bin_{(S+f)(P_0)}| = mr'$. The ratio of their average bin sizes for the schema $S+f$ is
$$\frac{\frac{\left[\sum_{S+f}size(bin_{(S+f)(i)})_{P_1}\right]}{mr'}}{\frac{\left[\sum_{S+f}size(bin_{(S+f)(i)})_{P_0}\right]}{mr'}} $$
The average bin size is directly correlated to the size of the sub-population for the same initial number of bins and the same dispersal ratio. Therefore, if one subgroup (say $P_0$) is smaller than the other ($P_1$), then the average bin size for $P_0$ is less than that of $P_1$. 

As the average bin size drops for members of a sub-population, the utility will also drop monotonically for partition-based algorithms.

% \documentclass{article}

% if you need to pass options to natbib, use, e.g.:
% \PassOptionsToPackage{numbers, compress}{natbib}
% before loading neurips_data_2023

% ready for submission
% \usepackage{neurips_data_2023}

% to compile a preprint version, add the [preprint] option, e.g.:
%     \usepackage[preprint]{neurips_data_2023}
% This will indicate that the work is currently under review.

% to compile a camera-ready version, add the [final] option, e.g.:
%     \usepackage[final]{neurips_data_2023}

% to avoid loading the natbib package, add option nonatbib:
%    \usepackage[nonatbib]{neurips_data_2023}

% Submissions to the datasets and benchmarks are typically non anonymous,
% but anonymous submissions are allowed. If you feel that you must submit 
% anonymously, you can compile an anonymous version by adding the [anonymous] 
% option, e.g.:
%     \usepackage[anonymous]{neurips_data_2023}
% This will hide all author names.

% \title{Extended Demonstration}

% \begin{document}

\section{Feature Definitions and Recommended Subsets}
Figure \ref{fig:features} lists the 24 Excerpts features.  The majority are from the 2019 American Community Survey Public Use Micodata; four of them (DENSITY, INDP\_CAT, EDU, PINCP\_DECILE) were derived from ACS features or public data as described in \ref{preprocessingcleaninglabeling}. Along with feature type, we've included cardinality (number of possible values). Because some deidentification algorithms require small feature spaces, the NIST CRC program recommends three smaller feature subsets: Demographic-focused, Industry-focused and Family-focused.  Each subset showcases different feature mechanics, while sharing common features to delineate subpopulations (SEX, MSP, RAC1P, OWN\_RENT, PINCP\_DECILE).

\begin{figure*}[h!]
\centering
    % \begin{subfigure}[t]{\textwidth}
        \centering
  % include first image
        \includegraphics[width=12cm]{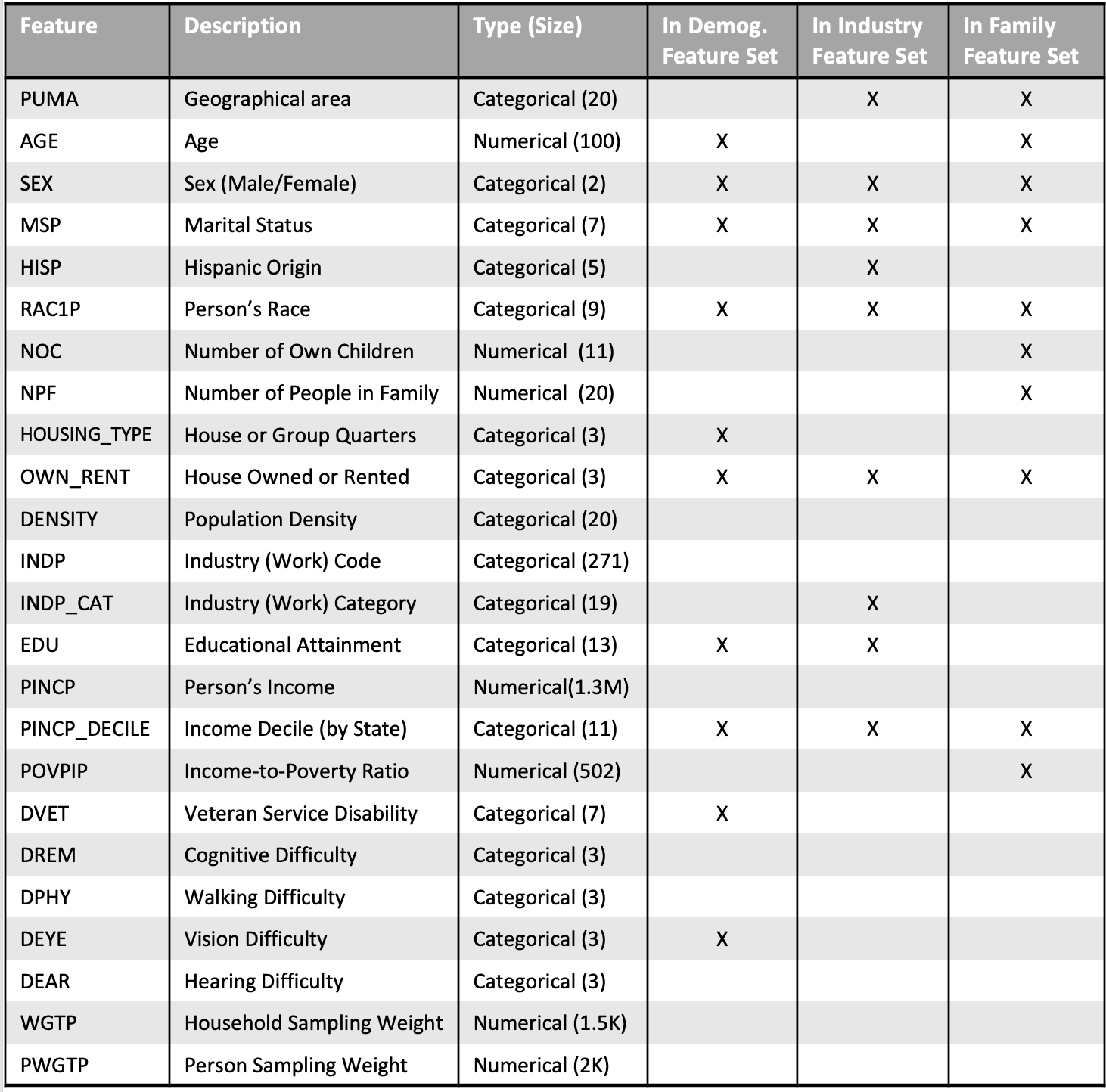}  
        \caption{The 24 Features in the Excerpts, and recommended feature subsets.}
        \label{fig:features}
    % \end{subfigure}
    \end{figure*}

\newpage
\section{Detailed Evaluation Reports and Metadata on Selected Deidentified Data Samples}
As we noted in the main paper, the \href{https://github.com/usnistgov/privacy_collaborative_research_cycle/}{NIST CRC Data and Metrics Bundle} is an archive of 300 deidentified data samples and evaluation metric results.  To demonstrate the efficacy of the Excerpts for identifying and diagnosing behaviors of deidentification algorithms on diverse populations, we selected seven algorithms from the archive to showcase in the paper.  Below we provide the complete meta-data and highlighted principal component analysis (PCA) plot for each sample, as well as links to their detailed evaluation reports (available online in the sample report section of the \href{https://github.com/usnistgov/SDNist/tree/main/sdnist/report/sample-reports}{SDNist repository}). 

Each detailed evaluation report contains the metrics listed below, along with complete results, detailed metric definitions accessible to non-technical stakeholders, a human-readable data dictionary, and additional references. 
\\

\textbf{SDNist Detailed Report Metrics List:} 
\begin{itemize}
\item K-marginal Edit Distance 
\item K-marginal Subsample Equivalent
\item K-marginal PUMA-specific Score 
\item Univariate Distribution Comparison
\item Kendall Tau Correlation Differences
\item Pearson Pairwise Correlation Differences 
\item Linear Regression (EDU vs PINCP\_DECILE), with Full 16 RACE + SEX Subpopulation Breakdowns
\item Propensity Distribution 
\item Pairwise Principle Component Analysis (Top 5) 
\item Pairwise PCA (Top 2, with MSP = `N' highlighting)
\item Inconsistencies (Age-based, Work-based, Housing-based)
\item Worst Performing PUMA Breakdown (Univariates and Correlations)
\item Privacy Evaluation: Unique Exact Match Metric
\item Privacy Evaluation: Apparent Match Metric 
\end{itemize}

\newpage
\subsection{Selected Algorithm Deidentified Data Summary Table}
For convenience, we include the deidentified data summary table from the main paper.  Expanded results for each algorithm are provided in sections E3-E10) 

\begin{table}[h!]
\centering
\resizebox{\textwidth}{!}{%
\begin{tabular}{|l|r|r|r|r|r|}
\hline
\\[-1em]
\multirow{2}{*}{} \textbf{Library and Algorithm } & \textbf{Privacy Type} & \textbf{Algorithm Type} & \textbf{Priv. Leak (UEM)} & \textbf{Utility (ES)} \\\hline
\\[-1em]
DP Histogram ($\epsilon = 10$) & differential privacy (DP) & simple histogram & 100 \% & $\sim 90$ \% (988)\\\hline
\\[-1em]
\multirow{2}{*}{}R synthpop CART model \cite{synthpop} & non-DP synthetic data & \shortstack{multiple imputation\\ decision tree} & 2.54 \% & $\sim$40 \% (935)\\\hline
\\[-1em]
\multirow{2}{*}{}MOSTLY AI SDG \cite{MostlyAI} \cite{platzer2022rule} & non-DP synthetic data & \shortstack{proprietary pre-trained\\ neural network} & 0.03 \% & $\sim$30 \% (921)\\\hline
\\[-1em]
\multirow{2}{*}{}SmartNoise MST ($\epsilon = 10$) \cite{mst} & DP & \shortstack{probabilistic graphical\\ model (PGM)} & 13.6 \% & $=10$\% (969)\\\hline
\\[-1em]
\multirow{2}{*}{}SDV CTGAN \cite{SDV} \cite{ctgan} & non-DP synthetic data & \shortstack{generative adversarial\\ network (GAN)} & 0.0 \% & $\sim$5 \% (775)\\\hline
\\[-1em]
SmartNoise PACSynth ($\epsilon = 10$) \cite{pacsynth} & DP + $k$-anonymity & constraint satisfaction & 0.87 \% & $\sim$1 \% (551)\\\hline
\\[-1em]
synthcity ADSGAN \cite{ads} \cite{synthcity}& custom noise injection & GAN & 0.0 \% & $<1$\% (121)\\\hline
\end{tabular}}
\caption{\label{tab:summary2}Summary of selected deidentification algorithms. Unique Exact Match (UEM) is a simple privacy metric that counts the percentage of singleton records in the target that are also present in the deidentified data; these uniquely identifiable individuals leaked through the deidentification process. The Equivalent Subsample (ES) utility metric uses an analogy between deidentification error and sampling error to communicate utility; a score of 5 \% indicates the edit distance between the target and deidentified data distributions is similar to the sampling error induced by randomly discarding 95 \% of the data. Edit distance is based on the k-marginal metric for sparse distributions. \cite{Abowd2021}, \cite{lothe_sdnist_2021}}
\label{deid_table}
\end{table}

\newpage

\subsection{Expanded Data for Subgroup Dispersal Line Graphs}
\begin{center}

\includegraphics[width=12cm]{figures/marginal_line_graph.png}
~

\includegraphics[width=12cm]{figures/privacy_line_graph.png}
\newpage

\begin{table}
    \begin{adjustwidth}{-10em}{}
    \centering
    \small    
%     \library & \algorithm & \epsilon & \subgrp & \4-marg. & \gap & \overall UEM ,overall percent UEM,overall 4-marg.
%     \csvreader[
%        head to column names
%      ]{figures/AppendixFigs/kmarg_demographic_uem_sorted.csv}{}{%
% \\\givenname\ \name & \matriculation }%
    \begin{filecontents*}{figures/AppendixFigs/kmarg_demographic_uem_sorted.csv}
,library,algorithm,epsilon,subgrp 4-marg. gap,overall UEM,overall percent UEM,overall 4-marg.
0,rsynthpop,ipf,1,100,526,3.53,836
1,subsample\_5pcnt,subsample\_5pcnt,,75,746,5.0,887
2,rsynthpop,ipf,2,156,1223,8.2,868
3,synthcity,tvae,,64,1279,8.57,799
4,smartnoise-synth,pacsynth,5,137,1521,10.34,817
5,rsynthpop,ipf,2,95,1593,10.68,898
6,smartnoise-synth,mst,1,108,1860,12.47,898
7,smartnoise-synth,pacsynth,10,145,2013,13.49,851
8,smartnoise-synth,mst,10,101,2027,13.59,906
9,smartnoise-synth,mst,5,95,2030,13.61,907
10,Sarus SDG,Sarus SDG,10,37,2087,13.99,905
11,smartnoise-synth,aim,1,77,2105,14.11,938
12,rsynthpop,ipf,10,157,2215,14.85,904
13,rsynthpop,ipf,100,155,2275,15.25,904
14,rsynthpop,ipf\_NonDP,,58,2379,15.82,958
15,smartnoise-synth,aim,5,48,2736,18.34,960
16,smartnoise-synth,aim,10,46,2789,18.7,964
17,synthcity,bayesian\_network,,12,3370,22.59,901
18,rsynthpop,cart,,25,3728,24.99,968
19,rsynthpop,cart,,78,5538,37.12,952
20,subsample\_40pcnt,subsample\_40pcnt,,23,5963,39.97,967
21,rsynthpop,catall,10,52,7066,47.37,902
22,sdcmicro,kanonymity,,74,7073,47.41,855
23,sdcmicro,pram,,29,8690,58.25,963
24,rsynthpop,catall,100,16,9351,62.68,973
25,rsynthpop,catall\_NonDP,,17,9453,63.37,973
26,sdcmicro,pram,,22,10160,68.11,968
27,rsynthpop,catall,10,38,10950,73.4,897
28,sdcmicro,kanonymity,,11,13219,88.61,960
29,rsynthpop,catall,100,1,14917,99.99,999
\end{filecontents*}
\csvautotabular{figures/AppendixFigs/kmarg_demographic_uem_sorted.csv}
    % \begin{tabular}{c|c}
    %      &  \\
    %      & 
    % \end{tabular}
    \end{adjustwidth}
    \caption{This table contains the specific deidentified data samples and metric results used to generate Figures 3 and 4 in the main paper.}
    \label{tab:my_label}
\end{table}

% \begin{figure*}
    
\end{center}

************************************************************
\newpage
\subsection{Differentially Private Histogram (epsilon-10)}
% ************************************************************
A differentially private histogram is a naive solution that simply counts the number of occurrences of each possible record value, and adds noise to the counts.  We use the Tumult Analytics library to efficiently produce a DPHistogram with a very large set of bins.  Epsilon 10 is a very weak privacy guarantee, and this simple algorithm provides very poor privacy in these conditions. The points in the 'deidentified' PCA are nearly the exact same points as in the target PCA. 

The full metric report can be found \href{https://htmlpreview.github.io/?https://github.com/usnistgov/SDNist/blob/main/sdnist/report/sample-reports/report_dphist_e_10_cf8_na2019_05-19-2023T18.01.12/report.html}{here}.

\begin{figure*}[h!]
    \begin{subfigure}[h!]{0.5\textwidth}
        \centering
  % include second image
        \includegraphics[width = \linewidth]{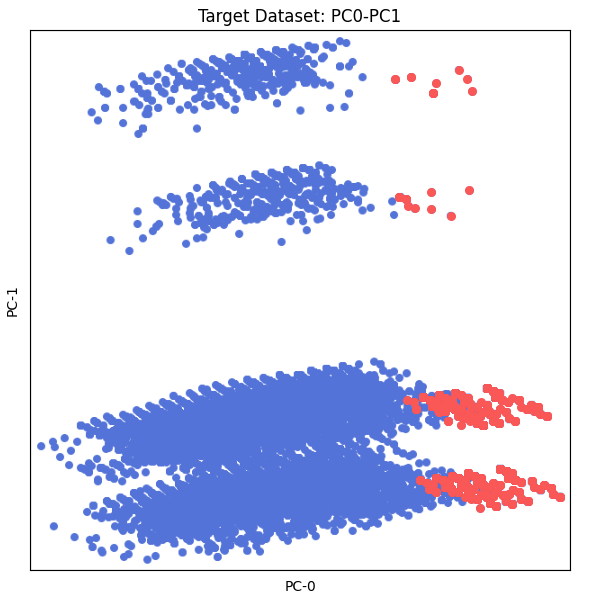} 
        \caption{Target data (Demographic-focused Feature Subset, excepting AGE and DEYE)}
        \label{Tumult1}
    \end{subfigure}
-
    \begin{subfigure}[h!]{0.5\textwidth}
        \centering
  % include second image
        \includegraphics[width = \linewidth]{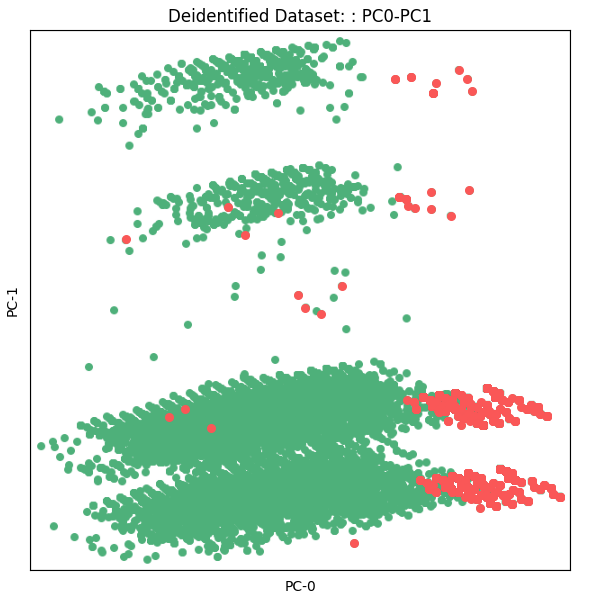}  
        \caption{Deidentified data}
        \label{Tumult2}
    \end{subfigure}
    \caption{The PCA Metric for DP Histogram ($\epsilon = 10$) }
\end{figure*}

\begin{table}[h!]
\centering
\begin{tabular}{l|l}
Label Name & Label Value\\\hline
Algorithm Name & DPHist\\
Library & Tumult Analytics\\
Feature Set & demographic-focused-except-AGEP-DEYE \\
Target Dataset & national2019\\
Epsilon & 10 \\
Privacy & differential privacy\\
Filename & dphist\_e\_10\_cf8\_na2019\\
Records & 27314\\
Features & 8\\
Library Link & \href{https://docs.tmlt.dev/analytics/latest/}{https://docs.tmlt.dev/analytics/latest/}
\end{tabular}
\caption{\label{tab:dphist-label}Label Information for Differential Private Histogram (epsilon-10)}
\end{table}

\newpage

% ************************************************************
\subsection{SmartNoise PACSynth (epsilon-10, Industry-focused)}
% ************************************************************
We've included two samples from the PACSynth library to showcase its behavior on different feature subsets.  The technique provides both differential privacy and a form of k-anonymity (removing rare outlier records).  This provides very good privacy, Table \ref{deid_table}, but it can also erase dispersed subpopulations. The industry feature subset below was used for the regression metric in the main paper, which showed erasure of graduate degree holders among both white men and black women.  

More information on the technique can be found \href{https://pages.nist.gov/privacy_collaborative_research_cycle/pages/techniques.html#smartnoise-pacsynth}{here}.  The full metric report can be found \href{ https://htmlpreview.github.io/?https://github.com/usnistgov/SDNist/blob/main/sdnist/report/sample-reports/report_pac_synth_e_10_industry_focused_na2019_05-19-2023T18.01.12/report.html}{here}.

\begin{figure*}[h!]
    \begin{subfigure}[h!]{0.5\textwidth}
        \centering
  % include second image
        \includegraphics[width = \linewidth]{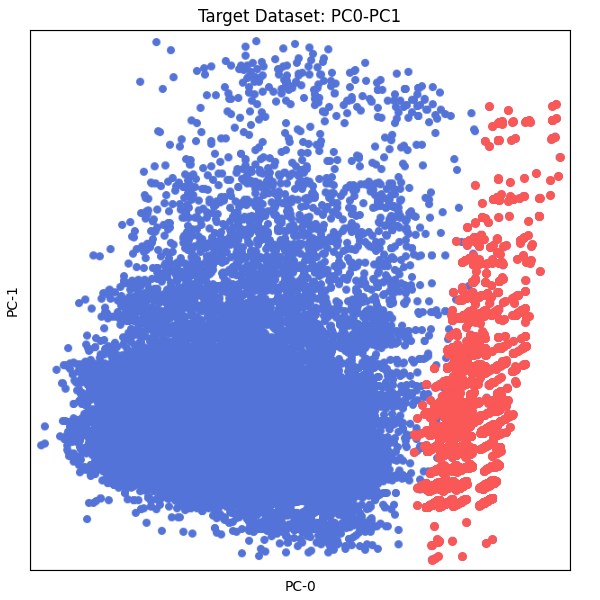} 
        \caption{Target data (Industry-focused Feature Subset)}
        \label{PACSynth1}
    \end{subfigure}
-
    \begin{subfigure}[h!]{0.5\textwidth}
        \centering
  % include second image
        \includegraphics[width = \linewidth]{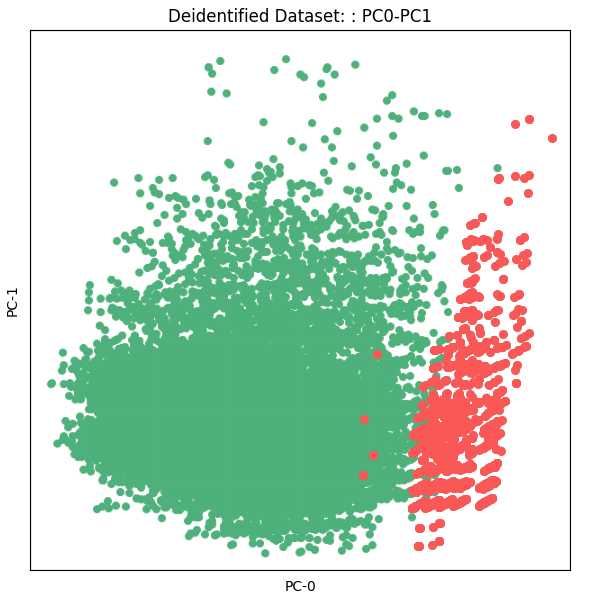}  
        \caption{Deidentified data}
        \label{PACSynth2}
    \end{subfigure}
    \caption{The PCA Metric for PACSynth ($\epsilon = 10$) }
\end{figure*}

\begin{table}[h!]
\centering
\begin{tabular}{l|l}
Label Name & Label Value\\\hline
Algorithm Name & pacsynth\\
Library & smartnoise-synth\\
Feature Set & industry-focused\\
Target Dataset & national2019\\
Epsilon & 10 \\
Variant Label & preprocessor-epsilon: 3\\
Privacy & Differential Privacy\\
Filename & pac\_synth\_e\_10\_industry\_focused\_na2019\\
Records & 29537\\
Features & 9\\
Library Link & \href{https://github.com/opendp/smartnoise-sdk/tree/main/synth}{https://github.com/opendp/smartnoise-sdk/tree/main/synth}
\end{tabular}
\caption{\label{tab:pacsynth-label}SmartNoise PACSynth (epsilon-10)}
\end{table}

\newpage
% ************************************************************
\subsection{SmartNoise PACSynth (epsilon-10), Family-focused)}
% ************************************************************
On the family-focused feature subset we can see the impact of the k-anonymity protection more dramatically.  Because the deidentified data with removed outliers have reduced diversity, it occupies a much smaller area in the plot as compared to the target data. The deidentified records are concentrated into fewer, more popular feature combinations and, thus, their points show less variance along the PCA axes. 

More information on the technique can be found \href{https://pages.nist.gov/privacy_collaborative_research_cycle/pages/techniques.html#smartnoise-pacsynth}{here}.
The full metric report can be found \href{https://htmlpreview.github.io/?https://github.com/usnistgov/SDNist/blob/main/sdnist/report/sample-reports/report_pac_synth_e_10_family_focused_na2019_05-19-2023T18.01.12/report.html}{here}.

\begin{figure*}[h!]
    \begin{subfigure}[h!]{0.5\textwidth}
        \centering
  % include second image
        \includegraphics[width = \linewidth]{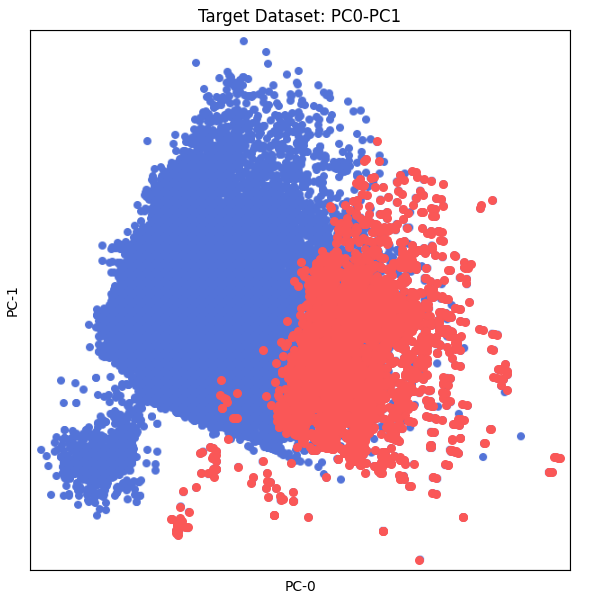} 
        \caption{Target data (Family-focused Feature Subset)}
        \label{PACSynth1a}
    \end{subfigure}
-
    \begin{subfigure}[h!]{0.5\textwidth}
        \centering
  % include second image
        \includegraphics[width = \linewidth]{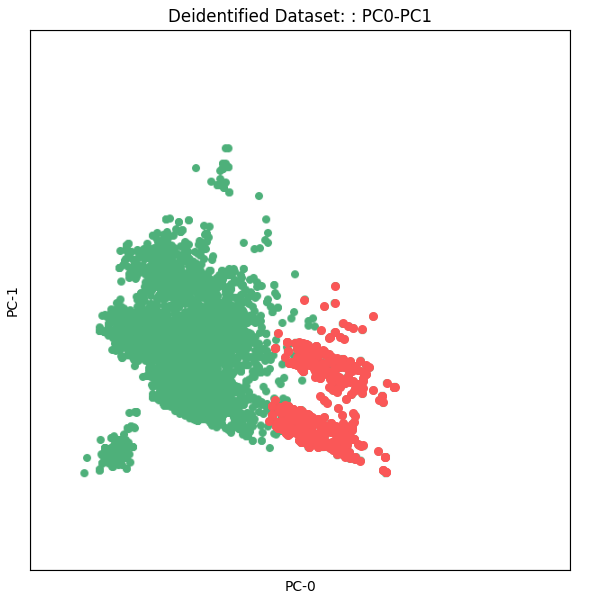}  
        \caption{Deidentified data}
        \label{PACSynth2a}
    \end{subfigure}
    \caption{The PCA Metric for PACSynth ($\epsilon = 10$) }
\end{figure*}

\begin{table}[h!]
\centering
\begin{tabular}{l|l}
Label Name & Label Value\\\hline
Algorithm Name & pacsynth\\
Library & smartnoise-synth\\
Feature Set & family-focused\\
Target Dataset & national2019\\
Epsilon & 10 \\
Variant Label & preprocessor-epsilon: 3\\
Privacy & differential privacy\\
Filename & pac\_synth\_e\_10\_industry\_focused\_na2019\\
Records & 29537\\
Features & 9\\
Library Link & \href{https://github.com/opendp/smartnoise-sdk/tree/main/synth}{https://github.com/opendp/smartnoise-sdk/tree/main/synth}
\end{tabular}
\caption{\label{tab:pacsynth-labela}SmartNoise PACSynth (epsilon-10)}
\end{table}

\newpage
% ************************************************************
\subsection{SmartNoise MST (epsilon-10)}
% ************************************************************.  
The MST synthesizer uses a probabilistic graphical model (PGM), with a maximum spanning tree (MST) structure capturing the most significant pair-wise feature correlations in the ground truth data as noisy marginal counts. This solution was the winner of the 2019 NIST Differential Privacy Synthetic Data Challenge.  Note that it provides good utility with much better privacy than the simple DP Histogram, but its selected marginals fail to capture some constraints on child records (in red). 

More information on the technique can be found \href{https://pages.nist.gov/privacy_collaborative_research_cycle/pages/techniques.html#smartnoise-mst}{here}.
The full metric report can be found \href{https://htmlpreview.github.io/?https://github.com/usnistgov/SDNist/blob/main/sdnist/report/sample-reports/report_mst_e10_demographic_focused_na2019_05-19-2023T18.01.12/report.html}{here}.

\begin{figure*}[h!]
    \begin{subfigure}[h!]{0.5\textwidth}
        \centering
  % include second image
        \includegraphics[width = \linewidth]{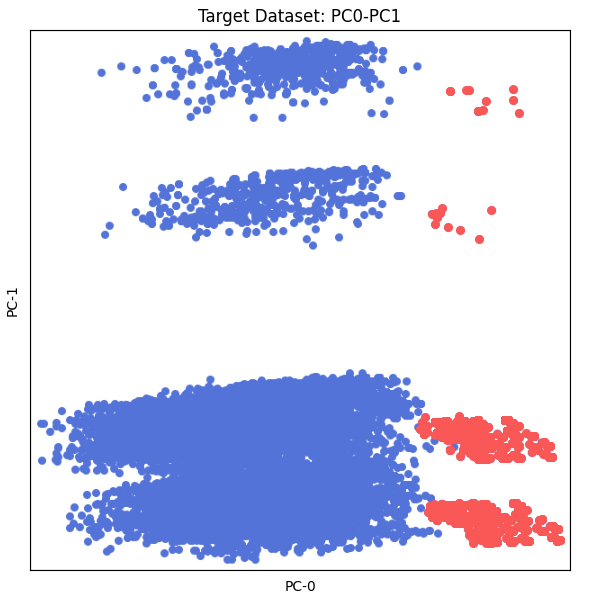} 
        \caption{Target data (Demographic-focused Feature Subset)}
        \label{MST1}
    \end{subfigure}
-
    \begin{subfigure}[h!]{0.5\textwidth}
        \centering
  % include second image
        \includegraphics[width = \linewidth]{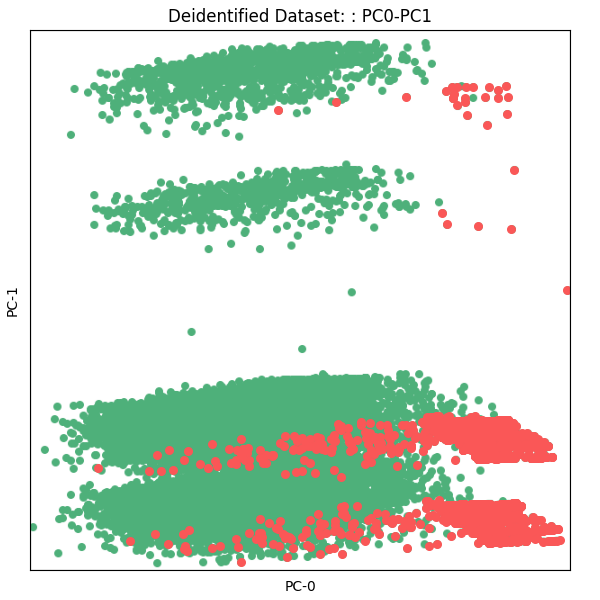}  
        \caption{Deidentified data}
        \label{MST2}
    \end{subfigure}
    \caption{The PCA Metric for MST ($\epsilon = 10$) }
\end{figure*}

\begin{table}[h!]
\centering
\begin{tabular}{l|l}
Label Name & Label Value\\\hline
Algorithm Name & mst\\
Library & smartnoise-synth\\
Feature Set & demographic-focused \\
Target Dataset & national2019\\
Epsilon & 10 \\
Variant Label & preprocessor-epsilon: 3\\
Privacy & differential privacy\\
Filename & mst\_e10\_demographic\_focused\_na2019\\
Records & 27253\\
Features & 10\\
Library Link & \href{https://github.com/opendp/smartnoise-sdk/tree/main/synth}{https://github.com/opendp/smartnoise-sdk/tree/main/synth}
\end{tabular}
\caption{\label{tab:mst-label}Label Information for SmartNoise MST (epsilon-10)}
\end{table}

\newpage
% ************************************************************
\subsection{R synthpop CART model}
% ************************************************************
The fully conditional Classification and Regression Tree (CART) model does not satisfy formal differential privacy, but provides better privacy than some techniques which do (Table \ref{deid_table}).  It uses a sequence of decision trees trained on the target data to predict each feature value based on the previously synthesized features; familiarity with decision trees is helpful for tuning this model. Note that the two PCA distributions have very similar shapes, comprised of different points. 

 You can find more information on the technique \href{https://pages.nist.gov/privacy_collaborative_research_cycle/pages/techniques.html#rsynthpop-cart}{here}.  The full metric report can be found \href{https://htmlpreview.github.io/?https://github.com/usnistgov/SDNist/blob/main/sdnist/report/sample-reports/report_cart_cf21_na2019_05-19-2023T18.01.12/report.html}{here}.

\begin{figure*}[h!]
    \begin{subfigure}[h!]{0.5\textwidth}
        \centering
  % include second image
        \includegraphics[width = \linewidth]{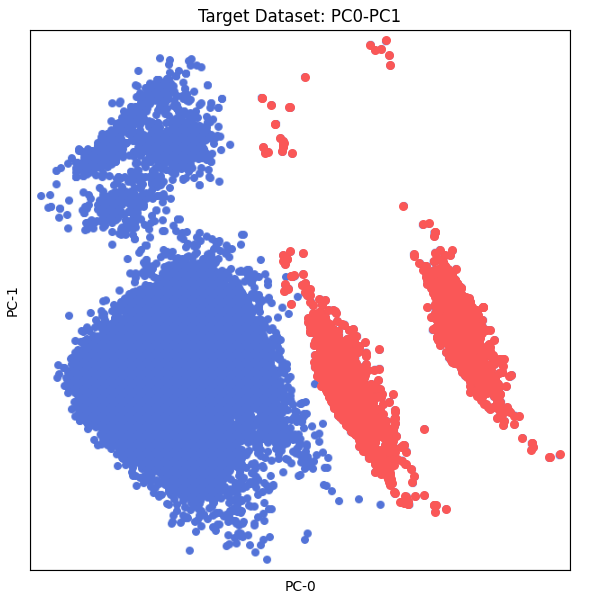} 
        \caption{Target data (21-Feature Subset)}
        \label{CART1}
    \end{subfigure}
-
    \begin{subfigure}[h!]{0.5\textwidth}
        \centering
  % include second image
        \includegraphics[width = \linewidth]{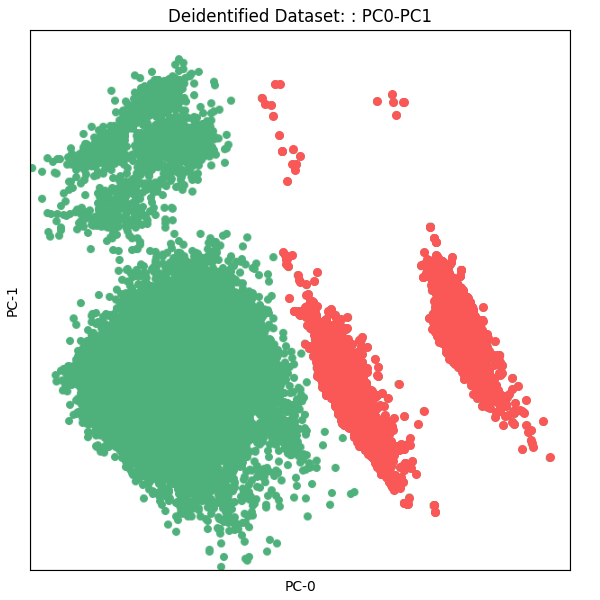}  
        \caption{Deidentified data}
        \label{CART2}
    \end{subfigure}
    \caption{The PCA Metric for CART}
\end{figure*}

\begin{table}[h!]
\centering
\begin{tabular}{l|l}
Label Name & Label Value\\\hline
Algorithm Name & cart\\
Library & rsynthpop\\
Feature Set & custom-features-21 \\
Target Dataset & national2019\\
Variant Label & maxfaclevels: 300\\
Privacy & Synthetic Data (Non-differentially Private)\\
Filename & cart\_cf21\_na2019\\
Records & 27253\\
Features & 21\\
Library Link & \href{https://cran.r-project.org/web/packages/synthpop/index.html}{https://cran.r-project.org/web/packages/synthpop/index.html}
\end{tabular}
\caption{\label{tab:rsynthpop-label}Label Information for R synthpop CART model}
\end{table}

\newpage
% ************************************************************
\subsection{MOSTLY AI Synthetic Data Platform}
% ************************************************************
MOSTLYAI is a proprietary synthetic data generation platform which uses a partly pretrained neural network model to generate  data. The model can be configured to respect deterministic constraints between features (for a comparison, see MOSTLYAI submissions 1 in the CRC Data and Metrics Bundle linked above). It does not provide differential privacy, but does very well on both privacy and utility metrics (Table \ref{deid_table}).   

More information on the technique can be found \href{https://pages.nist.gov/privacy_collaborative_research_cycle/pages/techniques.html#mostlyai-sd}{here}.  The full metric report can be found \href{https://htmlpreview.github.io/?https://github.com/usnistgov/SDNist/blob/main/sdnist/report/sample-reports/report_mostlyai_sd_platform_MichaelPlatzer_2_05-19-2023T18.01.12/report.html}{here}.

\begin{figure*}[h!]
    \begin{subfigure}[h!]{0.5\textwidth}
        \centering
  % include second image
        \includegraphics[width = \linewidth]{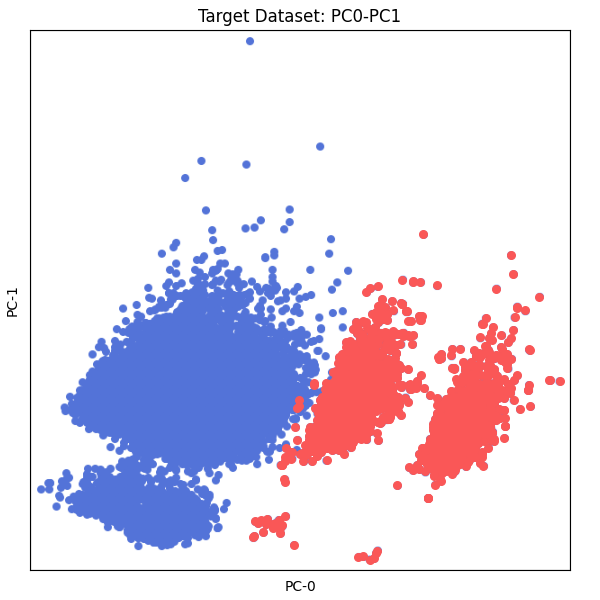} 
        \caption{Target data (All Features)}
        \label{MOSTLYAI1}
    \end{subfigure}
-
    \begin{subfigure}[h!]{0.5\textwidth}
        \centering
  % include second image
        \includegraphics[width = \linewidth]{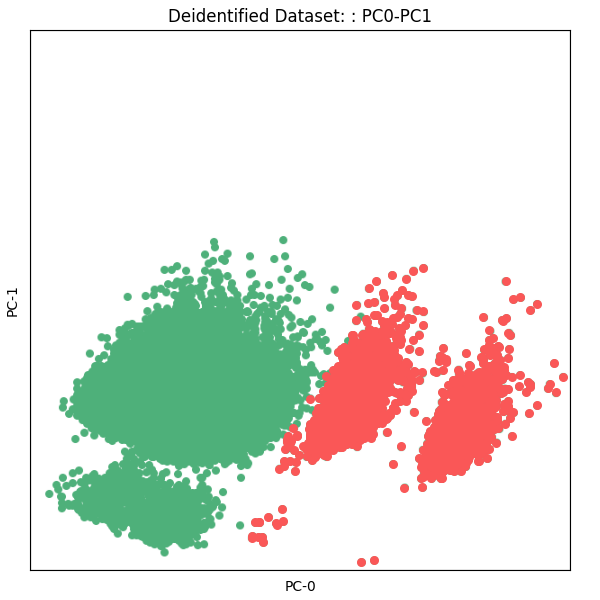}  
        \caption{Deidentified data}
        \label{MOSTLYAI2}
    \end{subfigure}
    \caption{The PCA Metric for MostlyAI}
\end{figure*}

\begin{table}[h!]
\centering
\begin{tabular}{l|l}
Label Name & Label Value\\\hline
Algorithm Name & MOSTLY AI\\
Submission Number & 2\\
Library & MostlyAI SD\\
Feature Set & all-features \\
Target Dataset & national2019\\
Variant Label & national2019\\
Privacy & Synthetic Data (Non-differentially Private)\\
Filename & mostlyai\_sd\_platform\_MichaelPlatzer\_2\\
Records & 27253\\
Features & 24\\
Library Link & \href{https://mostly.ai/synthetic-data}{https://mostly.ai/synthetic-data}
\end{tabular}
\caption{\label{tab:mostlyai-label}Label Information for MOSTLY AI Synthetic Data Platform}
\end{table}

\newpage
% ************************************************************
\subsection{Synthetic Data Vault CTGAN}
% ************************************************************
CTGAN is a type of Generative Adverserial Network designed to operate well on tabular data.  Unlike the MostlyAI neural network (which is pretrained with public data), the CTGAN network is only trained on the target data.  It is able to preserve some structure of the target data distribution, but it introduces artifacts.  In other metrics, we see it also has difficulty preserving diverse subpopulations. 
 
More information on the technique can be found \href{https://sdv.dev/SDV/user_guides/single_table/ctgan.html}{here}.
The full metric report can be found \href{https://htmlpreview.github.io/?https://github.com/usnistgov/SDNist/blob/main/sdnist/report/sample-reports/report_sdv_ctgan_epochs500_SlokomManel_1_05-19-2023T18.01.12/report.html}{here}.

\begin{figure*}[h!]
    \begin{subfigure}[h!]{0.5\textwidth}
        \centering
  % include second image
        \includegraphics[width = \linewidth]{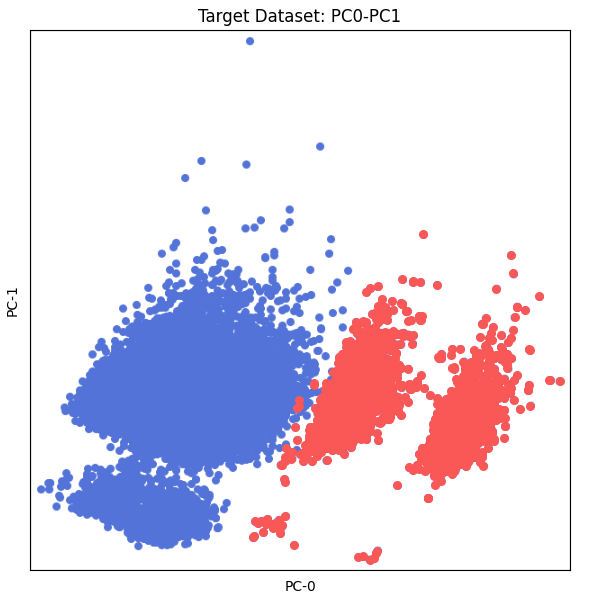} 
        \caption{Target data (All Features)}
        \label{CTGAN1}
    \end{subfigure}
-
    \begin{subfigure}[h!]{0.5\textwidth}
        \centering
  % include second image
        \includegraphics[width = \linewidth]{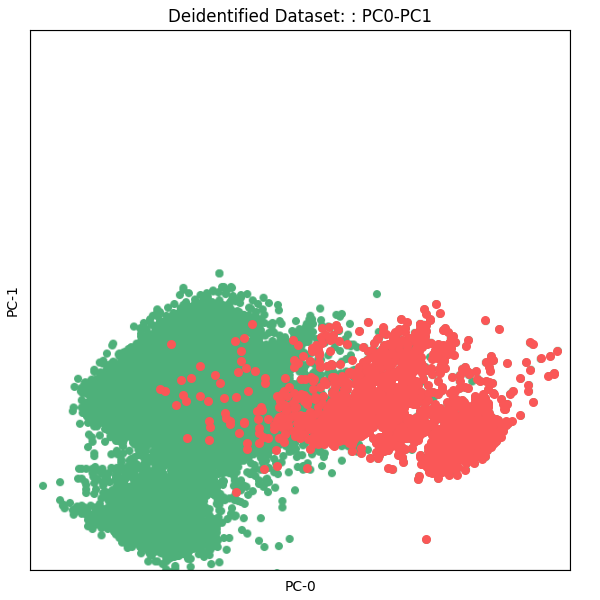}  
        \caption{Deidentified data}
        \label{CTGAN2}
    \end{subfigure}
    \caption{The PCA Metric for CTGAN}
\end{figure*}

\begin{table}[h!]
\centering
\begin{tabular}{l|l}
Label Name & Label Value\\\hline
Team & CBS-NL\\
Algorithm Name & ctgan\\
Submission Timestamp & 4/16/2023 12:03:58\\
Submission Number & 1\\
Library & sdv\\
Feature Set & all-features \\
Target Dataset & national2019\\
Variant Label & default CTGAN with epochs=500\\
Privacy & Synthetic Data (Non-differentially Private)\\
Filename & sdv\_ctgan\_epochs500\_SlokomManel\_1\\
Records & 27253\\
Features & 24\\
Library Link & \href{https://github.com/sdv-dev/CTGAN}{https://github.com/sdv-dev/CTGAN}
\end{tabular}
\caption{\label{tab:sdvctgan-label}Label Information for Synthetic Data Vault CTGAN}
\end{table}

\newpage
% ************************************************************
\subsection{synthcity ADSGAN}
% ************************************************************
ADSGAN is a Generative Adverserial Network focused on providing strong privacy for synthetic data.  While it doesn't formally satisfy differential privacy it uses a parameter alpha to inject noise during the training process.  Unfortunately, we see it is unable to preserve any meaningful structure from the target data distribution in this submission.  

More information on the technique can be found \href{https://pages.nist.gov/privacy_collaborative_research_cycle/pages/techniques.html#synthcity-adsgan}{here}.
The full metric report can be found \href{https://htmlpreview.github.io/?https://github.com/usnistgov/SDNist/blob/main/sdnist/report/sample-reports/report_adsgan_ZhaozhiQian_1_05-19-2023T18.01.12/report.html}{here}.

\begin{figure*}[h!]
    \begin{subfigure}[h!]{0.5\textwidth}
        \centering
  % include second image
        \includegraphics[width = \linewidth]{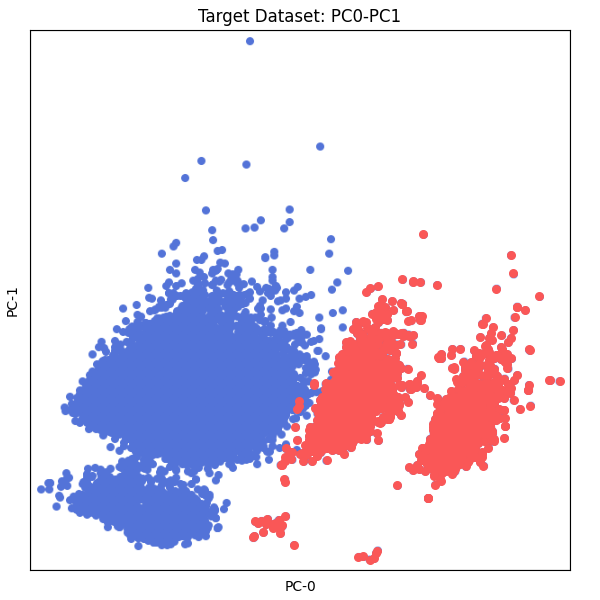} 
        \caption{Target data (All Features)}
        \label{CART1a}
    \end{subfigure}
-
    \begin{subfigure}[h!]{0.5\textwidth}
        \centering
  % include second image
        \includegraphics[width = \linewidth]{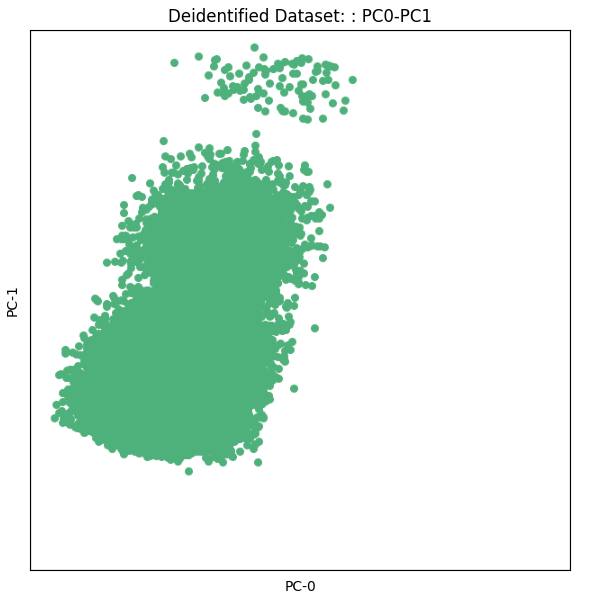}  
        \caption{Deidentified data}
        \label{CART2a}
    \end{subfigure}
    \caption{The PCA Metric for ADSGAN}
\end{figure*}

\begin{table}[h!]
\centering
\begin{tabular}{l|l}
Label Name & Label Value\\\hline
Team & CCAIM\\
Submission Timestamp & 3/9/2023 3:33:23\\
Submission Number & 1\\
Algorithm Name & adsgan\\
Library & synthcity\\
Feature Set & all-features\\
Target Dataset & national2019\\
Variant Label & default, lambda=10\\
Privacy & Synthetic Data (Non-differentially Private)\\
Filename & adsgan\_ZhaozhiQian\_1\\
Records & 21802\\
Features & 24\\
Library Link & \href{https://github.com/vanderschaarlab/synthcity}{https://github.com/vanderschaarlab/synthcity}
\end{tabular}
\caption{\label{tab:synthcityadsgan-label}Label Information for synthcity ADSGAN}
\end{table}

% \end{document}

\small
\bibliographystyle{unsrturl}
\bibliography{references}

\section{Feature Definitions and Properties}

Mention recommended feature subsets to make things accessible for guys that need smaller feature spaces, and also focus on different behavior.

Feature, description, type/cardinality,  demographic-focused, industry-focused, family-focused, simpler-21, full-features, challenges

\section{Selected SDNist Metrics}
In the main paper used this package to demonstrate the efficacy of the Excerpts data for identifying and diagnosing different deidentification algorithm behavior.  We provide a more detailed introduction to the metrics here, and in the next section we provide meta-data and expanded evaluation results on the selected deidentified data samples. 
 
\subsection{K-marginal and Equivalent Subsample}
\subsection{Unique Exact Match}
\subsection{Regression}
\subsection{Propensity}
\subsection{PCA}

\section{Detailed Evaluation Results on Selected Deidentified Data Samples}
In this section we provide meta-data and expanded evaluation results on the selected deidentified data samples.

\newpage
% ************************************************************
\subsection{Differentially Private Histogram (epsilon-10)}
% ************************************************************

\begin{figure*}[h!]
\centering
    % \begin{subfigure}[t]{\textwidth}
        \centering
  % include first image
        \includegraphics[width=\linewidth]{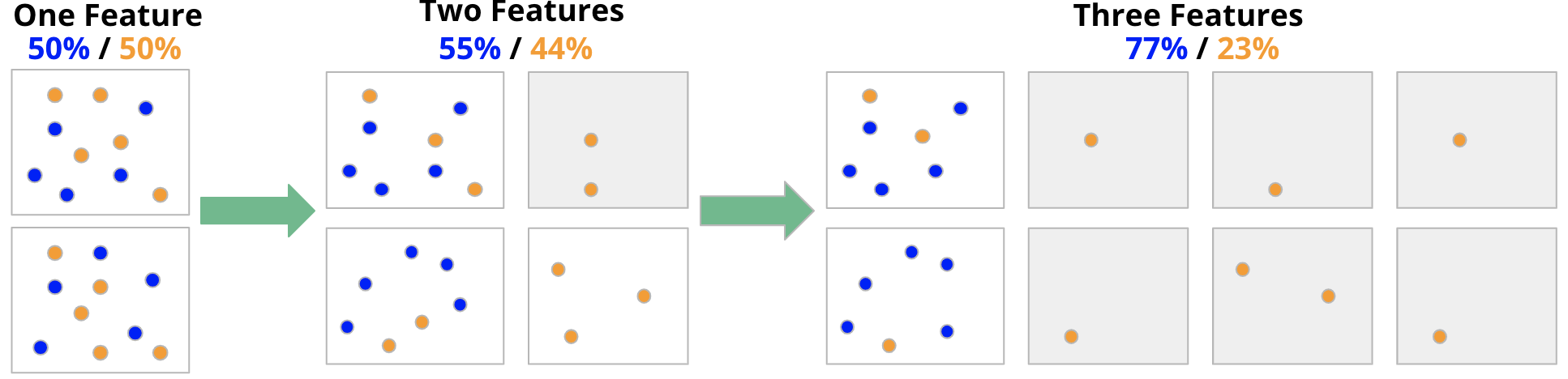}  
        \caption{}
        \label{fig:orange_blue}
    % \end{subfigure}
    \end{figure*}
~
\begin{figure*}
    \begin{subfigure}[t]{0.5\textwidth}
        \centering
  % include second image
        \includegraphics[width = 0.95 \linewidth]{figures/Dem_Dispersal.png} 
        \caption{Differing dispersal by demographic group in the Diverse Community Excerpts}
        \label{fig:dispersal_data}
    \end{subfigure}
~
    \begin{subfigure}[t]{0.5\textwidth}
        \centering
  % include second image
        \includegraphics[width = \linewidth]{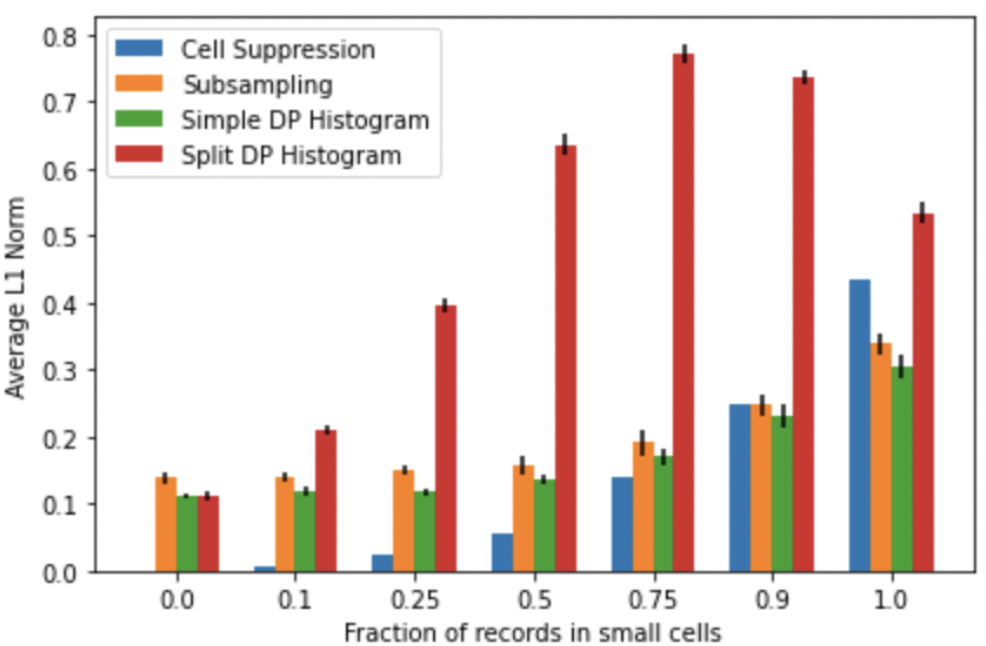}  
        \caption{Increasing error of deidentification by fraction of population dispersed in small count bins}
        \label{fig:small_bins}
    \end{subfigure}
    \caption{The equity dynamics of diverse data, dispersal, and deidentification}
\end{figure*}

\begin{table}
\centering
\begin{tabular}{l|l}
Label Name & Label Value\\\hline
Algorithm Name & DPHist\\
Library & tumult\\
Feature Set & demographic-focused-except-AGEP-DEYE \\
Target Dataset & national2019\\
Epsilon & 10 \\
Privacy & Differential Privacy\\
Filename & dphist\_e\_10\_cf8\_na2019\\
Records & 27314\\
Features & 8\\
Library Link & \hrf{https://docs.tmlt.dev/analytics/latest/}{https://docs.tmlt.dev/analytics/latest/}
\end{tabular}
\caption{\label{tab:dphist-label}Label Information for Differential Private Histogram (epsilon-10)}
\end{table}

\newpage
% ************************************************************
\subsection{SmartNoise PACSynth (epsilon-10)}
% ************************************************************

\begin{table}
\centering
\begin{tabular}{l|l}
Label Name & Label Value\\\hline
Algorithm Name & pacsynth\\
Library & smartnoise-synth\\
Feature Set & industry-focused\\
Target Dataset & national2019\\
Epsilon & 10 \\
Variant Label & preprocessor-epsilon: 3\\
Privacy & Differential Privacy\\
Filename & pac\_synth\_e\_10\_industry\_focused\_na2019\\
Records & 29537\\
Features & 9\\
Library Link & \href{https://github.com/opendp/smartnoise-sdk/tree/main/synth}{https://github.com/opendp/smartnoise-sdk/tree/main/synth}
\end{tabular}
\caption{\label{tab:pacsynth-label}SmartNoise PACSynth (epsilon-10)}
\end{table}

\newpage
% ************************************************************
\subsection{SmartNoise MST (epsilon-10)}
% ************************************************************

\begin{table}
\centering
\begin{tabular}{l|l}
Label Name & Label Value\\\hline
Algorithm Name & mst\\
Library & smartnoise-synth\\
Feature Set & demographic-focused \\
Target Dataset & national2019\\
Epsilon & 10 \\
Variant Label & preprocessor-epsilon: 3\\
Privacy & Differential Privacy\\
Filename & mst\_e10\_demographic\_focused\_na2019\\
Records & 27253\\
Features & 10\\
Library Link & \href{https://github.com/opendp/smartnoise-sdk/tree/main/synth}{https://github.com/opendp/smartnoise-sdk/tree/main/synth}
\end{tabular}
\caption{\label{tab:mst-label}Label Information for SmartNoise MST (epsilon-10)}
\end{table}

\newpage
% ************************************************************
\subsection{R synthpop CART model}
% ************************************************************

\begin{table}
\centering
\begin{tabular}{l|l}
Label Name & Label Value\\\hline
Algorithm Name & cart\\
Library & rsynthpop\\
Feature Set & custom-features-21 \\
Target Dataset & national2019\\
Variant Label & maxfaclevels: 300\\
Privacy & Synthetic Data (Non-differentially Private)\\
Filename & cart\_cf21\_na2019\\
Records & 27253\\
Features & 21\\
Library Link & \href{https://cran.r-project.org/web/packages/synthpop/index.html}{https://cran.r-project.org/web/packages/synthpop/index.html}
\end{tabular}
\caption{\label{tab:rsynthpop-label}Label Information for R synthpop CART model}
\end{table}

\newpage
% ************************************************************
\subsection{MOSTLY AI Synthetic Data Platform}
% ************************************************************

\begin{table}
\centering
\begin{tabular}{l|l}
Label Name & Label Value\\\hline
Algorithm Name & MOSTLY AI\\
Submission Number & 2\\
Library & MostlyAI SD\\
Feature Set & all-features \\
Target Dataset & national2019\\
Variant Label & national2019\\
Privacy & Synthetic Data (Non-differentially Private)\\
Filename & mostlyai\_sd\_platform\_MichaelPlatzer\_2\\
Records & 27253\\
Features & 24\\
Library Link & \href{https://mostly.ai/synthetic-data}{https://mostly.ai/synthetic-data}
\end{tabular}
\caption{\label{tab:mostlyai-label}Label Information for MOSTLY AI Synthetic Data Platform}
\end{table}

\newpage

% ************************************************************
\subsection{Synthetic Data Vault CTGAN}
% ************************************************************

\begin{table}
\centering
\begin{tabular}{l|l}
Label Name & Label Value\\\hline
Team & CBS-NL\\
Algorithm Name & ctgan\\
Submission Timestamp & 4/16/2023 12:03:58\\
Submission Number & 1\\
Library & sdv\\
Feature Set & all-features \\
Target Dataset & national2019\\
Variant Label & default CTGAN with epochs=500\\
Privacy & Synthetic Data (Non-differentially Private)\\
Filename & sdv\_ctgan\_epochs500\_SlokomManel\_1\\
Records & 27253\\
Features & 24\\
Library Link & \href{https://github.com/sdv-dev/CTGAN}{https://github.com/sdv-dev/CTGAN}
\end{tabular}
\caption{\label{tab:sdvctgan-label}Label Information for Synthetic Data Vault CTGAN}
\end{table}

\newpage
% ************************************************************
\subsection{synthcity ADSGAN}
% ************************************************************

\begin{table}
\centering
\begin{tabular}{l|l}
Label Name & Label Value\\\hline
Team & CCAIM\\
Submission Timestamp & 3/9/2023 3:33:23\\
Submission Number & 1\\
Algorithm Name & adsgan\\
Library & synthcity\\
Feature Set & all-features\\
Target Dataset & national2019\\
Variant Label & default, lambda=10\\
Privacy & Synthetic Data (Non-differentially Private)\\
Filename & adsgan\_ZhaozhiQian\_1\\
Records & 21802\\
Features & 24\\
Library Link & \href{https://github.com/vanderschaarlab/synthcity}{https://github.com/vanderschaarlab/synthcity}
\end{tabular}
\caption{\label{tab:synthcityadsgan-label}Label Information for synthcity ADSGAN}
\end{table}